\documentclass[11pt]{article}
\usepackage{amssymb}
\usepackage{amsmath}
\usepackage{amsthm}
\usepackage{graphics}
\usepackage[square,comma]{natbib}
\usepackage{epsfig,amssymb,latexsym,verbatim}
\usepackage{graphicx}
\usepackage{multirow}
\usepackage{multicol}
\usepackage{hypernat}
\usepackage{float}
\usepackage{hyperref}
\floatstyle{ruled}
\newfloat{algorithm}{tbp}{loa}
\floatname{algorithm}{Algorithm}
\usepackage{graphicx, amssymb}

\newtheorem{theorem}{{\bf Theorem}}
\newtheorem{lemma}{{\bf Lemma}}

\begin{document}

\renewcommand{\baselinestretch}{1.2}
\markboth{\hfill{\footnotesize\rm LI WANG}\hfill} {\hfill
{\footnotesize\rm Semiparametric Survey Sampling} \hfill}
\renewcommand{\thefootnote}{}
$\ $\par \fontsize{10.95}{14pt plus.8pt minus .6pt}\selectfont
\vspace{0.8pc} \centerline{\Large\bf Single-Index Model-Assisted
Estimation In Survey Sampling}
\vspace{.4cm} \centerline{Li Wang \footnote{\emph{Address for
correspondence}: Li Wang, Department of Statistics, University of
Georgia, Athens, Georgia, 30602, USA. E-mail: lilywang@uga.edu }}
\vspace{.4cm} \centerline{\it University of Georgia}
\vspace{.55cm} \fontsize{9}{11.5pt plus.8pt minus .6pt}\selectfont

\begin{quotation}
\noindent \textit{Abstract:} A model-assisted semiparametric
method of estimating finite population totals is investigated to
improve the precision of survey estimators by incorporating
multivariate auxiliary information. The proposed superpopulation
model is a single-index model which has proven to be a simple and
efficient semiparametric tool in multivariate regression. A class
of estimators based on polynomial spline regression is proposed.
These estimators are robust against deviation from single-index
models. Under standard design conditions, the proposed estimators
are asymptotically design-unbiased, consistent and asymptotically
normal. An iterative optimization routine is provided that is
sufficiently fast for users to analyze large and complex survey
data within seconds. The proposed method has been applied to
simulated datasets and MU281 dataset, which have provided strong
evidence that corroborates with the asymptotic theory.

\vspace{9pt} \noindent \textit{Key words and phrases:}
Horvitz-Thompson estimator; model-assisted estimation;
semiparametric; spline smoothing; superpopulation.
\end{quotation}

\fontsize{10.95}{14pt plus.8pt minus .6pt}\selectfont

\thispagestyle{empty}

\section{Introduction}
\renewcommand{\thetheorem}{{\sc \arabic{theorem}}} \renewcommand{%
\theproposition}{{\sc \thesection.\arabic{proposition}}} \renewcommand{%
\thelemma}{{\sc \thesection.\arabic{lemma}}} \renewcommand{\thecorollary}{{%
\sc \thesection.\arabic{corollary}}} \renewcommand{\theequation}{%
\thesection.\arabic{equation}} \renewcommand{\thesubsection}{\thesection.%
\arabic{subsection}} \setcounter{equation}{0} \setcounter{lemma}{0} %
\setcounter{proposition}{0} \setcounter{theorem}{0} %
\setcounter{subsection}{0}\setcounter{corollary}{0}

\label{SEC:introduction}

In this article, the classic finite-population estimation problem
is investigated. In what follows, let $U_{N}=\left\{
1,...,i,...,N\right\} $ denote the $N$ units of finite population.
For each $i\in U_{N}$, let $y_{i}$
be a generic characteristic and the objective is to estimate $%
t_{y}=\sum_{i\in U_{N}}y_{i}$. A probability sample $s$ is drawn
from $U_{N}$
according to a fixed sampling design $p_{N}\left( \cdot \right) $, where $%
p_{N}\left( s\right) $ is the probability of drawing the sample $s$. Let $%
\pi _{iN}\equiv \pi _{i}=\Pr \left\{ i\in s\right\} =\sum_{s\ni
i}p_{N}\left( s\right) $ denote the inclusion probability for
element $i\in U_{N}$ and $\pi _{ijN}\equiv \pi _{ij}=\Pr \left\{
i,j\in s\right\} =\sum_{s \ni i,j}p_{N}\left( s\right) $ denote
the inclusion probability for element $i,j\in U_{N}$.

If no information other than the inclusion probabilities is used
to estimate $t_{y}$, a well-known design unbiased estimator is the
Horvitz-Thompson estimator
\begin{equation}
\hat{t}_{y\pi }\equiv \hat{t}_{y}=\sum_{i\in s}\frac{y_{i}}{\pi
_{i}}. \label{DEF:HT}
\end{equation}
The variance of the Horvitz-Thompson estimator under the sampling
design is
\[
\text{Var}_{p}\left( \hat{t}_{y}\right) =\sum_{i,j\in U_{N}}\left(
\pi _{ij}-\pi _{i}\pi _{j}\right) \frac{y_{i}}{\pi
_{i}}\frac{y_{j}}{\pi _{j}}.
\]

The efficiency of the Horvitz-Thompson estimator can be
significantly improved by incorporating some ``cheap'' auxiliary
information at the population level in addition to sample data.
Such auxiliary information is often available for all elements of
the population of interest in many surveys. For instance, in many
countries, administrative registers provide extensive sources of
auxiliary information. Complete registers can give access to
variables such as sex, age, income and country of birth. Studies
of labor force characteristics or household expenditure patterns,
for example, might benefit from these auxiliary data. Another
example is the satellite images or GPS data used in spatial
sampling. These data are often collected at the population level,
which are often available at little or no extra cost, especially
compared to the cost of collecting the survey data. For more
examples of auxiliary information, see
\citep{REF:C,REF:CDW,REF:D,REF:DH}. Use of auxiliary information
to improve the accuracy of survey estimators actually dates back
to post-stratification, calibration, ratio and regression
estimation; see \citep{REF:CDH,REF:CS,REF:SSW,REF:T} for a general
review of these methods. Auxiliary information can also be used to
increase the accuracy of the finite population distribution
function, for example, \citep{REF:WD}.

In this article, let $\mathbf{x}_{i}=\left\{
x_{i1},...,x_{id}\right\} $ be a $d$-dimensional auxiliary
variable vector, $i\in U_{N}$, and assume that $\left\{ \left( \mathbf{x}%
_{i},y_{i}\right) \right\} _{i\in U_{N}}$ is a realization of
$\left( \mathbf{X},Y\right) $ from an infinite superpopulation,
$\xi $, satisfying
\begin{equation}
Y=m\left( \mathbf{X}\right) +\sigma \left( \mathbf{X}\right)
\varepsilon , \label{model}
\end{equation}
in which the $d$-variate function $m$ is the unknown mean function
of $Y$ conditional on the auxiliary information vector $\mathbf{X}
$, often is assumed to be smooth; $\sigma$ is the unknown standard
deviation function. The standard error satisfies that $E_{\xi
}\left( \varepsilon \left\vert \mathbf{X}\right. \right) =0$ and
$E_{\xi }\left( \varepsilon ^{2}\left\vert \mathbf{X}\right.
\right) =1$, where $E_{\xi }$ is the expectation with respect to
the population $\xi $. The interesting problem is how to take
advantage of the regression relationship (\ref{model}) to better estimate $%
t_{y}$.

The traditional parametric approach to analyze a regression
relationship assumes that the superpopulation model is fully
described by a finite set of parameters, for example, the linear
regression estimator discussed in \citep{REF:SSW}. However, it
sometimes requires prohibitively complex models with a very large
number of parameters to address various hypotheses. It is very
difficult to obtain any prior model information about the
regression function $m$ in (\ref{model}), and substantial
estimation bias can result if a preselected parametric model is
too restricted to fit unexpected features. As an alternative one
can try to estimate the unknown regression relationships
nonparametrically without reference to a specific form. The
flexibility of nonparametric smoothing/regression is extremely
helpful in exploratory data analysis as well as in obtaining
robust predictions, see \citep{REF:FG,REF:HW} for details.

Nonparametric methods for survey data are rather sparse and have
begun to emerge as important and practical tools, see
\citep{REF:BO,REF:BCO,REF:MR,REF:OBMK,REF:ZL03,REF:ZL05}.
Reference \citep{REF:BO} first proposed a nonparametric
model-assisted estimator based on local polynomial regression,
which generalized the parametric framework in survey sampling and
improved the precision of the survey estimators. Their
investigation is restricted to the scalar case, i.e., $d=1$.
Nowadays most surveys involve more than one auxiliary variable
(reference \citep{REF:SL}). For example, the auxiliary information
obtained from remote sensing data, satellite images and GPS data
provide a wide and growing range of variables to be employed in
spatial sampling. Northeastern lakes survey discussed in
\citep{REF:BOJR,REF:ES} is a good example of this. In that study,
a lot of information, such as longitude, latitude, and elevation,
of every lake in the population is known for the Environmental
Monitoring and Assessment Program (EMAP) of the U.S. Environmental
Protection Agency (EPA). In addition, the growing possibilities of
information and communication technology have made it possible to
develop very large and complex surveys. In this article, a
$d$-dimensional auxiliary vector is considered to improve the
efficiency of estimating $t_{y}$ for both small and large surveys.

Research in nonparametric survey theory and methodology when the
dimension of the auxiliary information vector is high, however, is
quite challenging. A key difficulty is due to the issue of ``curse
of dimensionality'': the optimal rate of convergence decreases
with dimensionality (\citep{REF:Stone}). One solution is
regression in the form of additive model popularized by
\citep{REF:HT}; see \citep{REF:BCO,REF:BOJR,REF:OBMK} for possible
application of additive model to survey sampling. A weakness of
the purely additive model is that interactions between the
explanatory variables are completely ignored (\citep{REF:STY}). An
attractive alternative to additive model is the single-index model
given in (\ref{model:SIM}). Similar to the first step of
projection pursuit regression, single-index model reduces
dimensionality but does not incorporate interactions; see
\citep{REF:CFGW,REF:HP,REF:HHI,REF:HH,REF:XTLZ} for instance. The
basic appeal of single-index model is that it is in nature a
hybrid method of parametric and nonparametric regression. It
preserves the simplicity of parametric regression where simplicity
is sufficient: the $d$-variate function $m(\mathbf{x})=m\left(
x_{1},...,x_{d}\right) $ is expressed as a
univariate function of $\mathbf{x}^{T}\mathbf{\theta }_{0}=%
\sum_{q=1}^{d}x_{q}\theta _{0,q}$; it also employs the flexibility
of nonparametric regression where flexibility is necessary.

In this article, I investigate the single-index model-assisted
estimator for the finite population total, that is, the
superpopulation model in (\ref{model}) is assumed to be a SIM.
Under standard design conditions, a design-consistent estimator of
$\mathbf{\theta }_{0}$ has been obtained using polynomial splines,
and the proposed estimator of $t_{y}$ is asymptotically
design-unbiased, consistent and asymptotically normal. By taking
advantage of the spline smoothing and iterative optimization
routines, the proposed method is particularly computationally
efficient comparing to the kernel additive model approaches in the
literature of nonparametric survey estimation, in which iterative
approaches such as a backfitting algorithm
(\citep{REF:BOJR,REF:HT}) or marginal integration (\citep{REF:LN})
are necessary. The rest of the article is organized as follows.
Section \ref {SEC:estimator} gives details of the model
specification and the proposed method of estimation. Section
\ref{SEC:results} describes some nice properties of the estimator.
Section \ref{SEC:algorithm} provides the actual procedure to
implement the method. Section \ref{SEC:empirical results} reports
the empirical results. All technical proofs are contained in
Appendices A and B.

\section{Superpopulation Model and Proposed Estimator}

\label{SEC:estimator}

\renewcommand{\thetheorem}{{\sc \arabic{theorem}}} \renewcommand{%
\theproposition}{{\sc \thesection.\arabic{proposition}}} \renewcommand{%
\thelemma}{{\sc \thesection.\arabic{lemma}}} \renewcommand{\thecorollary}{{%
\sc \thesection.\arabic{corollary}}} \renewcommand{\theequation}{%
\thesection.\arabic{equation}} \renewcommand{\thesubsection}{\thesection.%
\arabic{subsection}} \setcounter{equation}{0} \setcounter{lemma}{0} %
\setcounter{proposition}{0} \setcounter{theorem}{0} %
\setcounter{subsection}{0}\setcounter{corollary}{0}

\subsection{Single-index superpopulation model} \label{SUBSEC:superpopulation}

In this article, the proposed superpopulation model $\xi$ in
(\ref{model}) is a single-index model (SIM), where
\begin{equation}
Y=m\left(\mathbf{X}^{T}\mathbf{\theta }_{0}\right)+\sigma \left(
\mathbf{X}\right)\varepsilon, \label{model:SIM}
\end{equation}
where the unknown parameter $\mathbf{\theta }_{0}$ is called the
single-index coefficient, used for simple interpretation once
estimated; function $m$ is an unknown smooth function used for
further data summary.

If the SIM is misspecified, however, a goodness-of-fit test is
necessary and the estimation of $\mathbf{\theta }_{0}$ must be
rethought; see \citep{REF:XLTZ}. So in this article, instead of
presuming that the underlying true function $m$ is a single-index
function like the one defined in (\ref{model:SIM}), the
single-index is identified by the best approximation to the
multivariate function $m$. Specifically, a univariate function $g$
is estimated that optimally approximates the multivariate function
$m$ in the sense of
\begin{equation}
g\left( \nu \right) =E_{\xi }\left[ \left. m\left(
\mathbf{X}\right) \right| \mathbf{X}^{T}\mathbf{\theta }_{0}=\nu
\right]. \label{DEF:g}
\end{equation}
The superiority of this method is that it works very well even
under model misspecification so that it is much more useful in
applications than the traditional SIMs given in (\ref{model:SIM}).

For the superpopulation model defined by (\ref{model}) and
(\ref{DEF:g}), let
\[
m_{\mathbf{\theta }} \left( \mathbf{X}^{T}\mathbf{\theta }\right)
=E_{\xi }\left[ Y|\mathbf{X}^{T}\mathbf{\theta }\right] =E_{\xi
}\left[ \left. m\left( \mathbf{X} \right) \right\vert
\mathbf{X}^{T}\mathbf{\theta }\right]
\]
for any fixed $\mathbf{\theta }$, where as noted in the introduction, $%
E_{\xi }$ denotes the expected value with respect to the
population $\xi $ in (\ref{model}) and (\ref{DEF:g}). Define the
risk function of $\mathbf{\theta }$ as
\begin{equation}
R\left( \mathbf{\theta }\right) =E_{\xi }\left[ \left\{ Y-m_{\mathbf{\theta }%
}\left( \mathbf{X}^{T}\mathbf{\theta }\right) \right\} ^{2}\right]
=E_{\xi }\left\{
m\left( \mathbf{X}\right) -m_{\mathbf{\theta }}\left( \mathbf{X}^{T}\mathbf{%
\theta }\right) \right\} ^{2}+E_{\xi }\sigma ^{2}\left(
\mathbf{X}\right) , \label{DEF:Rtheta}
\end{equation}
which is uniquely minimized at $\mathbf{\theta }_{0}\in
S_{+}^{d-1}=\left\{ \left( \theta _{1},...,\theta _{d}\right)
|\sum_{q=1}^{d}\theta _{q}^{2}=1,\theta _{d}>0\right\} $.

\noindent \textbf{Remark 2.1:} It is obvious that without
constraints, the coefficient vector $\mathbf{\theta }_{0}$ is
identified only up to a constant factor. Typically, one requires
that $\left\Vert \mathbf{\theta }_{0}\right\Vert =1$ which entails
that at least one of the coordinates $\theta _{0,1},...,\theta
_{0,d}$ is nonzero. One could assume without loss of generality
that $\theta _{0,d}>0$, and the candidate $\mathbf{\theta }_{0}$
would then belong to the upper unit hemisphere $S_{+}^{d-1}$.

\subsection{Spline smoothing}

Estimation of both $\mathbf{\theta }_{0}$ and $g(\cdot)$ in model
(\ref{DEF:g}) requires a degree of statistical smoothing. In this
article, all estimation is carried out via polynomial splines. The
use of polynomial spline smoothing in the generalized
nonparametric models can be back to \citep{REF:Stone}. As pointed
out in \citep{REF:BCO,REF:WY}, one of the important advantages of
spline smoothing is the relative ease with which spline estimators
can be simply computed, even for large datasets or datasets with
regions of sparse data. In addition, spline smoothing is a global
smoothing method. After the spline basis is chosen, the
coefficients can be estimated by an efficient optimization
procedure. In contrast, kernel based methods such as the kernel
based backfitting (\citep{REF:BOJR,REF:HT}) and marginal
integration approaches (\citep{REF:LN}), in which the maximizing
has to be conducted repeatedly at every local data points, are
very time-consuming.

To introduce the function space of splines of order $p$, one
pre-selects an integer $N^{1/6}\ll J=J_{N}\ll N^{1/5}\left( \log
N\right)
^{-2/5}$, see Assumption (A4) below, and divides $\left[ 0,1\right] $ into $%
\left( J+1\right) $ subintervals, $\left[ k_{j},k_{j+1}\right) $, $%
j=0,...,J-1 $, $\left[ k_{J},1\right] $, where $\left\{
k_{j}\right\} _{j=1}^{J}$ is a sequence of equally-spaced points,
called interior knots, given as
\[
k_{1-p}=...=k_{-1}=k_{0}=0<k_{1}<...<k_{J}<1=k_{J+1}=...=k_{J+p},
\]
in which $k_{j}=j/(J+1)$, $j=0,1,...,J+1$. The $j$-th B-spline of
order $p$ denoted by $B_{j,p}$ is recursively defined by
\citep{REF:de}. In the following, let $\Phi ^{\left( 2\right)
}=\Phi ^{\left( 2\right) }\left[ 0,1\right] $ be the space of all
the second order smoothness functions that are polynomials of
degree $3$ on each subinterval.

Direct calculation shows that under Assumption (A1) in Section
\ref{SUBSEC:assumptions}, for any $\mathbf{\theta }\in
S_{+}^{d-1}$, the variable $\mathbf{X}^{T}\mathbf{\theta }$ has a
Lebesgue probability density function (pdf) that is uniformly
bounded below and above by the pdf of a rescaled centered
$\text{Beta}\left\{ \left( d+1\right) /2,\left( d+1\right)
/2\right\} $,
\[
f_{d}\left( \nu \right) =\frac{\Gamma \left( d+1\right) }{\Gamma
\left\{ \left( d+1\right) /2\right\} ^{2}2^{d}a}\left(
1-v^{2}/a^{2}\right) ^{\left( d-1\right) /2}I_{\left[ -a,a\right]
}\left( v\right),
\]
which vanishes at boundary points $-a$ and $a$. This makes
nonparametric smoothing of $Y$ on $\mathbf{X}^{T}\mathbf{\theta }$
difficult. I therefore first transform the variable
$\mathbf{X}^{T}\mathbf{\theta }$ by using the
cumulative distribution function $F_{d}$ of $f_{d}$%
\begin{equation}
F_{d}\left( \nu \right) =\int_{-1}^{\nu /a}\frac{\Gamma \left( d+1\right) }{%
\Gamma \left\{ \left( d+1\right) /2\right\} ^{2}2^{d}}\left(
1-t^{2}\right) ^{\left( d-1\right) /2}dt,\nu \in \left[
-a,a\right] .  \label{DEF:Fd}
\end{equation}

For the rest of the article, denote the transformed variable of
the single-index variable $X^{T}{\mathbf{\theta }}$ by
$Z_{\mathbf{\theta }}$ and let $\varphi _{\mathbf{\theta }}$ be
the conditional expectation of $m$ given the transformed variable
$Z_{\mathbf{\theta }}$, i.e.
\begin{equation}
\varphi _{\mathbf{\theta }}\left( Z_{\mathbf{\theta }}\right)
=E_{\xi}\left\{ m\left( \mathbf{X}\right) |Z_{\mathbf{\theta
}}\right\} =E_{\xi}\left\{ m\left( \mathbf{X}\right)
|X^{T}{\mathbf{\theta }}\right\} =m_{\mathbf{\theta }}\left(
X^{T}{\mathbf{\theta }}\right) , \label{EQ:phitheta-mtheta}
\end{equation}

\noindent \textbf{Remark 2.2:} The transformed variable,
$Z_{\mathbf{\theta }}$, has a quasi-uniform $[0,1]$ distribution,
i.e., the pdf of the transformed variable is supported on $\left[
0,1\right] $ with positive lower bound. In practice, the radius
$a$ can take the value of the $100(1-\alpha)$ percentile of
$\left\{ \Vert \mathbf{x}_{i}\Vert \right\} _{i\in U_{N}}$, for
example, $\alpha=0.05$.

\subsection{``Oracle'' population-based estimator} \label{SUBSEC:population-based}

If the entire realization were known
by ``oracle'', one can create an ``oracle'' estimator to estimate $\mathbf{%
\theta }_{0}$ and $g$ in (\ref{DEF:g}) through a profile
least-squares method. One first estimates the single-index
coefficient $\mathbf{\theta }_{0}$ by a consistent estimator
$\tilde{\mathbf{\theta }}$ via minimizing the empirical version of
the risk function $R\left( \mathbf{\theta }\right)$ defined in
(\ref{DEF:Rtheta}), i.e.
\begin{equation}
\tilde{\mathbf{\theta }}=\arg \min_{\mathbf{\theta }\in S_{+}^{d-1}}\tilde{R}%
\left( \mathbf{\theta }\right) .  \label{DEF:thetatilde}
\end{equation}
where
\begin{equation}
\tilde{R}\left( \mathbf{\theta }\right) =N^{-1}\sum_{i\in
U_{N}}\left\{ y_{i}-\tilde{\varphi}_{\mathbf{\theta }}\left(
z_{\mathbf{\theta }i}\right) \right\} ^{2},  \label{DEF:Rtilde}
\end{equation}
and
\begin{equation}
\tilde{\varphi}_{\mathbf{\theta }}\left( \cdot \right) =\arg
\min_{\varphi \left( \cdot \right) \in \Phi ^{\left( 2\right)
}\left[ 0,1\right] }\sum_{i\in U_{N}}\left\{ y_{i}-\varphi \left(
z_{\mathbf{\theta }i}\right) \right\} ^{2}.
\label{DEF:mthetatilde}
\end{equation}
Then the link function $g$ can be estimated by $\tilde{g}$, a
cubic spline
smoother of $\left\{y_{i}\right\}$ on $\left\{z_{\tilde{\mathbf{\theta }}%
i}\right\}$, i.e., $\tilde{g}\left( \nu \right) =\tilde{\varphi}_{\mathbf{%
\tilde{\theta}}}\left( F_{d}\left( \nu \right) \right) $, where
$F_{d}\left( \cdot \right) $ is defined in (\ref{DEF:Fd}). Thus
the best single-index approximation to $m(\mathbf{x})$ is
$\tilde{m}(\mathbf{x})=\tilde{g}\left(
\mathbf{x}^{T}\mathbf{\tilde{\theta}}\right)=\tilde{\varphi}_{\mathbf{\tilde{%
\theta}}}\left( z_{\mathbf{\tilde{\theta}}}\right)$.

Let $\mathbf{y}=\left( y_{1},y_{2},...,y_{N}\right) ^{T}$, $\mathbf{B}_{%
\mathbf{\theta }}=\left\{ B_{j,4}\left( z_{\mathbf{\theta
}i}\right)
\right\} _{i\in U_{N},j=-3,...,J}$ be the B-spline matrix for any fixed $%
\mathbf{\theta }$ and $\mathbf{e}_{i}$ be a $N$-vector with a $1$ in the $i$%
th position and $0$ elsewhere. Write
\begin{equation}
\tilde{m}_{i}=\tilde{g}\left( \mathbf{x}_{i}^{T}\tilde{\mathbf{\theta}}%
\right) =\tilde{\varphi}_{\tilde{\mathbf{\theta}}}\left( z_{\tilde{\mathbf{%
\theta}}i}\right) =\mathbf{e}_{i}^{T}\mathbf{B}_{\tilde{\mathbf{\theta }}%
}\left( \mathbf{B}_{\tilde{\mathbf{\theta }}}^{T}\mathbf{B}_{\tilde{\mathbf{%
\theta }}}\right) ^{-1}\mathbf{B}_{\tilde{\mathbf{\theta
}}}^{T}\mathbf{y.} \label{DEF:mitilde}
\end{equation}
Clearly, $\tilde{m}_{i}$ is the spline single-index prediction at $\mathbf{x}%
_{i}$ based on the entire finite population. If these pseudo predictions $%
\tilde{m}_{i}$ were known, then a design-unbiased estimator of
$t_{y}$ would be the generalized difference estimator
\begin{equation}
\tilde{t}_{y,\text{diff}}=\sum_{i\in s}\frac{y_{i}-\tilde{m}_{i}}{\pi _{i}}%
+\sum_{i\in U_{N}}\tilde{m}_{i},  \label{DEF:tydifftilde}
\end{equation}
as given on page 221 of reference \citep{REF:SSW}. The design
variance of $\tilde{t}_{y,\text{diff}}$ in (\ref{DEF:tydifftilde})
is
\[
\text{Var}_{p}\left( \tilde{t}_{y,\text{diff}}\right)
=\sum_{i,j\in
U_{N}}\left( \pi _{ij}-\pi _{i}\pi _{j}\right) \frac{y_{i}-\tilde{m}_{i}}{%
\pi _{i}}\frac{y_{j}-\tilde{m}_{j}}{\pi _{j}}.
\]

\subsection{Sample-based estimator}\label{SUBSEC:sample-based}

However, the predictions $\tilde{m}_{i}$ for $m\left(
\mathbf{x}_{i}\right) $ can not be computed directly from data,
because the only $y_{i}$'s observed are those with $i\in s$.
Therefore, each $\tilde{m}_{i}$ needs to be replaced by a
sample-based consistent estimator. For any fixed $\mathbf{\theta
}$, the
sample based cubic spline estimator $\hat{\varphi}_{\mathbf{\theta }}$ of $%
\varphi _{\mathbf{\theta }}$ in (\ref{EQ:phitheta-mtheta}) is
defined as
\begin{equation}
\hat{\varphi}_{\mathbf{\theta }}\left( \cdot \right) =\arg
\min_{\varphi \left( \cdot \right) \in \Phi ^{\left( 2\right)
}\left[ 0,1\right]
}\sum_{i\in s}\pi _{i}^{-1}\left\{ y_{i}-\varphi \left( z_{\mathbf{\theta }%
i}\right) \right\} ^{2}.  \label{DEF:mthetahat}
\end{equation}
Define the sample-based empirical risk function of $\mathbf{\theta
}$
\begin{equation}
\hat{R}\left( \mathbf{\theta }\right) =N^{-1}\sum_{i\in s}\pi
_{i}^{-1}\left\{ y_{i}-\hat{\varphi}_{\mathbf{\theta }}\left( z_{\mathbf{%
\theta }i}\right) \right\} ^{2},  \label{DEF:Rhat}
\end{equation}
then the sample design-based spline estimator of $\mathbf{\theta
}_{0}$ is defined as
\begin{equation}
\hat{\mathbf{\theta }}=\arg \min_{\mathbf{\theta \in }S_{+}^{d-1}}\hat{R}%
\left( \mathbf{\theta }\right) ,  \label{DEF:thetahat}
\end{equation}
and the spline estimator of $g$ is $\hat{g}$, i.e., $\hat{g}\left(
v\right) =\hat{\varphi}_{\mathbf{\hat{\theta}}}\left( F_{d}\left(
v\right) \right) $. For any $i\in s$, let
\begin{equation}
\hat{m}_{i}=\hat{g}\left( \mathbf{x}_{i}^{T}\hat{\mathbf{\theta }}\right) =%
\hat{\varphi}_{\mathbf{\hat{\theta}}}\left( z_{\mathbf{\hat{\theta}}%
i}\right) =\mathbf{e}_{i}^{T}\mathbf{B}_{\hat{\mathbf{\theta
}},s}\left(
\mathbf{B}_{\hat{\mathbf{\theta }},s}^{T}\mathbf{W}_{s}\mathbf{B}_{\hat{%
\mathbf{\theta }},s}\right) ^{-1}\mathbf{B}_{\hat{\mathbf{\theta }},s}^{T}%
\mathbf{W}_{s}\mathbf{y}_{s},  \label{DEF:mihat}
\end{equation}
where $\mathbf{y}_{s}=\left\{ y_{i}\right\} _{i\in s}$ is the
$n_{N}$-vector of $y_{i}$ obtained in the sample and
\[
\mathbf{B}_{\hat{\mathbf{\theta }},s}=\left\{ B_{j,4}\left( z_{\mathbf{\hat{%
\theta}}i}\right) \right\} _{i\in s,j=-3,...,J},\mathbf{W}_{s}=\text{diag}%
\left\{ \frac{1}{\pi _{i}}\right\} _{i\in s}.
\]
Then the sample design-based B-spline estimator of $t_{y}$ is
\begin{equation}
\hat{t}_{y,\text{diff}}=\sum_{i\in s}\frac{y_{i}-\hat{m}_{i}}{\pi _{i}}%
+\sum_{i\in U_{N}}\hat{m}_{i}. \label{DEF:tydiffhat}
\end{equation}

\section{Properties of the estimator}

\label{SEC:results}

\renewcommand{\thetheorem}{{\sc \arabic{theorem}}} \renewcommand{%
\theproposition}{{\sc \thesection.\arabic{proposition}}} \renewcommand{%
\thelemma}{{\sc \thesection.\arabic{lemma}}} \renewcommand{\thecorollary}{{%
\sc \thesection.\arabic{corollary}}} \renewcommand{\theequation}{%
\thesection.\arabic{equation}} \renewcommand{\thesubsection}{\thesection.%
\arabic{subsection}} \setcounter{equation}{0} \setcounter{lemma}{0} %
\setcounter{proposition}{0} \setcounter{theorem}{0} %
\setcounter{subsection}{0}\setcounter{corollary}{0}

\subsection{A simple alternative expression for the estimator}
\label{SUBSEC:expression}

Like the ratio and linear regression estimators (\citep{REF:SSW})
and the penalized spline estimators (\citep{REF:BCO}), the
B-spline estimator defined in (\ref{DEF:tydiffhat}) can also be
represented in a simple form. Let $\mathbf{\hat{t}}_{z}$ and
$\mathbf{\hat{t}}_{z\pi }$ be two vectors:
\[
\mathbf{\hat{t}}_{z}=\left\{ \sum_{i\in U_{N}}B_{j,4}\left( z_{%
\mathbf{\hat{\theta}}i}\right) \right\} _{j=-3,...,J}^{T}, \text{
} \mathbf{\hat{t}}_{z\pi }=\left\{ \sum_{i\in s}\pi_{i}^{-1}
B_{j,4}\left( z_{\mathbf{\hat{\theta}}i}\right) \right\}
_{j=-3,...,J}^{T}.
\]
Then the estimator in (\ref{DEF:tydiffhat}) can be written as $\hat{t}_{y,%
\text{diff}}=\hat{t}_{y}+\left( \mathbf{\hat{t}}_{z}-\mathbf{\hat{t}}%
_{z\pi }\right) \hat{\gamma}$, where
\[
\hat{\gamma}=\left( \mathbf{B}_{\mathbf{\hat{\theta}},s}^{T}\mathbf{W}_{s}%
\mathbf{B}_{\mathbf{\hat{\theta}},s}\right) ^{-1}\mathbf{B}_{\mathbf{\hat{%
\theta}},s}^{T}\mathbf{W}_{s}\mathbf{y}_{s}.
\]
Noting that $\left( 1,...,1\right) _{J+4}\mathbf{B}_{\mathbf{\hat{\theta}}%
,s}^{T}=\left( 1,...,1\right) _{n_{N}}$ and
\[
\left\{ \sum_{i\in s}\pi_{i}^{-1}B_{j,4}\left( z_{\mathbf{\hat{%
\theta}}i}\right) \right\} _{j=-3,...,J}^{T}=\left( 1,...,1\right)
_{J+4}\mathbf{B}_{\mathbf{\hat{\theta}},s}^{T}\mathbf{W}_{s}\mathbf{B}_{%
\mathbf{\hat{\theta}},s},
\]
one has
\begin{eqnarray*}
\mathbf{\hat{t}}_{z\pi }\hat{\gamma} &=&\left\{ \sum_{i\in
s}\pi_{i}^{-1}B_{j,4}\left( z_{\mathbf{\hat{\theta}}i}\right)
\right\}
_{j=-3,...,J}^{T}\left( \mathbf{B}_{\mathbf{\hat{\theta}},s}^{T}%
\mathbf{W}_{s}\mathbf{B}_{\mathbf{\hat{\theta}},s}\right) ^{-1}\mathbf{B}_{%
\mathbf{\hat{\theta}},s}^{T}\mathbf{W}_{s}\mathbf{y}_{s} \\
&=&\left( 1,...,1\right) _{J+4}\mathbf{B}_{\mathbf{\hat{\theta}},s}^{T}%
\mathbf{W}_{s}\mathbf{y}_{s}=\left( 1,...,1\right) _{n_{N}}\mathbf{W}_{s}%
\mathbf{y}_{s}=\hat{t}_{y}.
\end{eqnarray*}
So the proposed estimator takes the simple and attractive form:
$\hat{t}_{y,\text{diff}}=\mathbf{\hat{t}}_{z}\hat{\gamma}=\sum_{i\in
U_{N}}\hat{m}_{i} $.

\subsection{Assumptions}
\label{SUBSEC:assumptions}

I will use the traditional asymptotic framework given in reference
\citep{REF:BO,REF:IF}, in which both the population and sample
sizes increase as $N\rightarrow \infty$. There are two sources of
``variation'' to be considered here. The first is introduced by
the random sample design and the corresponding measure is denoted
by $p$. The ``with $p$-probability 1'', ``$O_{p}$'', ``$o_{p}$''
and ``$E_{p}(\cdot)$'' notation below is with respect to this
measure. The second is associated with the superpopulation from
which the finite population is viewed as a sample. The
corresponding measure and notation are ``$\xi$'', ``with
$\xi$-probability 1'', ``$O_{\xi}$'', ``$o_{\xi}$'' and
``$E_{\xi}(\cdot)$''.

Before stating the asymptotic properties of the estimators, we
need some assumptions. Let $B_{a}^{d}=\left\{ \mathbf{x}\in R^{d}\left\vert \left\Vert \mathbf{x}%
\right\Vert \leq a\right. \right\} $ be the $d$-dimensional ball
with radius $a$, center $\mathbf{0}$ and volume
$\text{Vol}_{d}\left( B_{a}^{d}\right)$. Let
\[
C^{(k)}\left( B_{a}^{d}\right) =\left\{ m\left\vert \text{the\
}k\text{-th order partial derivatives of }m\text{ are continuous
on }B_{a}^{d}\right. \right\}
\]
be the space of $k$-th order smooth functions. Before stating the
asymptotic results, I formulate some assumptions:
\begin{enumerate}
\item[(A1)]  \textit{The density function of }$\mathbf{X}$,\textit{\ }$%
f\left( \mathbf{x}\right) \in C^{(4)}\left( B_{a}^{d}\right)
$\textit{\ for
some }$a>0$\textit{, and there are positive constants }$c_{f}\leq C_{f}$%
\textit{\ such that }$c_{f}/\text{Vol}_{d}\left( B_{a}^{d}\right)
\leq
f\left(\mathbf{x}\right) \leq C_{f}/\text{Vol}_{d}\left( B_{a}^{d}\right) $%
\textit{, if }$\mathbf{x}\in B_{a}^{d}$\textit{\ and }$f\left( \mathbf{x}%
\right) =0,\mathbf{x}\notin B_{a}^{d}$.

\item[(A2)]  \textit{The regression function in} (\ref{model})
$m\in C^{(4)}\left( B_{a}^{d}\right) $.

\item[(A3)]  \textit{The error} $\varepsilon $ \textit{in
(\ref{model}) satisfies} $E_{\xi }\left( \varepsilon \left|
\mathbf{X}\right. \right) =0$,
$E_{\xi }\left( \varepsilon ^{2}\left| \mathbf{X}\right. \right) =1$ \textit{%
and there exists a positive constant} $M$ \textit{such that} $\sup_{\mathbf{x%
}\in B_{a}^{d}}E_{\xi }\left( \left| \varepsilon \right| ^{3}\left| \mathbf{X%
}=\mathbf{x}\right. \right) <M$. \textit{The standard deviation function }$%
\sigma \left( \mathbf{x}\right) $ \textit{is continuous on} $B_{a}^{d}$, $%
0<c_{\sigma }\leq \inf_{\mathbf{x}\in B_{a}^{d}}\sigma \left( \mathbf{x}%
\right) \leq \sup_{\mathbf{x}\in B_{a}^{d}}\sigma \left(
\mathbf{x}\right) \leq C_{\sigma }<\infty $.

\item[(A4)]  \textit{As }$N\rightarrow \infty $,
$n_{N}N^{-1}\rightarrow \pi \in \left( 0,1\right) $ \textit{and
the number of interior knots} $J_{N}$ \textit{satisfies}:
$n_{N}^{1/6}\ll J_{N}\ll n_{N}^{1/5}\left\{ \log \left(
n_{N}\right) \right\} ^{-2/5}$.

\item[(A5)]  \textit{For all} $N$, $\min_{i\in U_{N}}\pi _{i}\geq
\lambda
>0, $ $\min_{i,j\in U_{N}}\pi _{ij}\geq \lambda ^{*}>0$ \textit{and}
\[
\lim \sup_{N\rightarrow \infty }n_{N}\max_{i,j\in U_{N},i\neq
j}\left| \pi _{ij}-\pi _{i}\pi _{j}\right| <\infty .
\]

\item[(A6)]  \textit{Let} $D_{k,N}$ \textit{be the set of all distinct} $k$%
\textit{-tuples} $\left( i_{1},i_{2},...,i_{k}\right) $
\textit{from} $U_{N}$,
\[
\lim \sup_{N\rightarrow \infty }n_{N}^{2}\max_{\left(
i_{1},i_{2},i_{3},i_{4}\right) \in D_{4,N}}\left| E_{p}\left[
\left( I_{i_{1}}-\pi _{i_{1}}\right) \left( I_{i_{2}}-\pi
_{i_{2}}\right) \left( I_{i_{3}}-\pi _{i_{3}}\right) \left(
I_{i_{4}}-\pi _{i_{4}}\right) \right] \right| <\infty ,
\]
\[
\lim \sup_{N\rightarrow \infty }n_{N}^{2}\max_{\left(
i_{1},i_{2},i_{3},i_{4}\right) \in D_{4,N}}\left| E_{p}\left[
\left( I_{i_{1}}I_{i_{2}}-\pi _{i_{1}i_{2}}\right) \left(
I_{i_{3}}I_{i_{4}}-\pi _{i_{3}i_{4}}\right) \right] \right|
<\infty ,
\]
\textit{and}
\[
\lim \sup_{N\rightarrow \infty }n_{N}^{2}\max_{\left(
i_{1},i_{2},i_{3}\right) \in D_{3,N}}\left| E_{p}\left[ \left(
I_{i_{1}}-\pi _{i_{1}}\right) ^{2}\left( I_{i_{2}}-\pi
_{i_{2}}\right) \left( I_{i_{3}}-\pi _{i_{3}}\right) \right]
\right| <\infty,
\]
\textit{where $I_{i}=1$ if $i\in s$ and $I_{i}=0$ otherwise}.

\item[(A7)] \textit{The risk function }$\tilde{R}$ \textit{in
(\ref
{DEF:Rtilde}) is locally convex at} $\tilde{\mathbf{\theta }}$\textit{: }$%
\forall \varepsilon >0,\exists \delta >0$ \textit{such that
}$\left\|
\mathbf{\theta }-\tilde{\mathbf{\theta }}\right\| _{2}<\varepsilon $ \textit{%
if }$\tilde{R}\left( \mathbf{\theta }\right) -\tilde{R}\left( \tilde{\mathbf{%
\theta }}\right) <\delta $.

\item[(A8)] \textit{The second order partial derivative of the
risk function, $\tilde{R}\left( \mathbf{\theta }\right)$, is
bounded at $\mathbf{\theta }=\tilde{\mathbf{\theta }}$.}
\end{enumerate}

\noindent \textbf{Remark 3.1:} Assumptions (A1)-(A3) are typical
in the nonparametric smoothing literature, see for instance,
\citep{REF:FG,REF:HW,REF:XTLZ}. Assumptions (A4) is about how to
choose the number of knots in order to achieve the optimal
nonparametric rate of convergence. In practice the number of
interior knots $J_{N}$ is chosen according to
(\ref{EQ:numberofknots}). Assumptions (A5) and (A6) involve the
inclusion probabilities of the design, which are also assumed in
reference \citep{REF:BO}. Assumption (A7) is used to derive the
design consistency of $\hat{\mathbf{\theta}}$ to
$\tilde{\mathbf{\theta}}$ and Assumption (A8) is used to obtain
the rate of the consistency.

\subsection{Asymptotic properties of the estimator} \label{SUBSEC:properties}

The estimator $\hat{\mathbf{\theta}}$ in (\ref{DEF:thetahat}) of
the single-index coefficient $\mathbf{\theta}_{0}$ is
asymptotically design consistent as the following theorem
demonstrates.

\begin{theorem}
\label{THM:theta consistent} Under Assumptions (A1)-(A5) and (A7), $%
\hat{\mathbf{\theta }}$ is asymptotically design consistent in the
sense that with $p$-probability 1
\[
\lim_{N\rightarrow \infty }\hat{\mathbf{\theta }}-\tilde{\mathbf{\theta }}%
=0,
\]
and further if (A8) holds, then
\[
\left\|\hat{\mathbf{\theta }}-\tilde{\mathbf{\theta
}}\right\|_{\infty}=O_p\left(J_{N}/n_N^{1/2}\right),
\]
where $\tilde{\mathbf{\theta }}$ and $\hat{\mathbf{\theta }}$ are
the population and sample based estimators of $\mathbf{\theta
}_{0}$ in (\ref {DEF:thetatilde}) and (\ref{DEF:thetahat}).
\end{theorem}

Like the local polynomial estimators in \citep{REF:BO}, the
following theorem shows that the estimator
$\hat{t}_{y,\text{diff}}$ in (\ref{DEF:tydiffhat}) is
asymptotically design unbiased and design consistent.

\begin{theorem}
\label{THM:ADU} Under Assumptions (A1)-(A5) and (A7)-(A8), the
model assisted spline estimator $\hat{t}_{y,\text{diff}}$ in
(\ref{DEF:tydiffhat}) is asymptotically design unbiased (ADU) in
the sense that
\[
\lim_{N\rightarrow \infty }E_{p}\left[ \frac{\hat{t}_{y,\text{diff}}-t_{y}%
}{N}\right] =0\text{ with }\xi \text{-probability }1,
\]
and is design consistent in the sense that for all $\eta >0$
\[
\lim_{N\rightarrow \infty }E_{p}\left[ I_{\left\{ \left| \hat{t}_{y,\text{%
diff}}-t_{y}\right| >N\eta \right\} }\right] =0\text{ with }\xi \text{%
-probability }1.
\]
\end{theorem}

Like the local polynomial estimators in \citep{REF:BO}, the
following theorem shows that the estimator in
(\ref{DEF:tydiffhat}) also inherits the limiting distribution of
the generalized difference estimator.

\begin{theorem}
\label{THM:normality}Under Assumptions (A1)-(A8), for $\tilde{t}_{y,\text{%
diff}}$ and $\hat{t}_{y,\text{diff}}$ in (\ref{DEF:tydifftilde})
and (\ref {DEF:tydiffhat}),
\[
\frac{N^{-1}\left( \tilde{t}_{y,\text{diff}}-t_{y}\right) }{\textrm{Var}%
_{p}^{1/2}\left( N^{-1}\tilde{t}_{y,\text{diff}}\right) }\stackrel{d}{%
\longrightarrow }N\left( 0,1\right)
\]
as $N\rightarrow \infty $ implies
\[
\frac{N^{-1}\left( \hat{t}_{y,\text{diff}}-t_{y}\right) }{\hat{V}%
^{1/2}\left( N^{-1}\hat{t}_{y,\text{diff}}\right) }\stackrel{d}{%
\longrightarrow }N\left( 0,1\right) ,
\]
where
\begin{equation}
\hat{V}\left( N^{-1}\hat{t}_{y,\text{diff}}\right) =\frac{1}{N^{2}}%
\sum_{i,j\in U_{N}}\left( y_{i}-\hat{m}_{i}\right) \left( y_{j}-\hat{m}%
_{j}\right) \left( \frac{\pi _{ij}-\pi _{i}\pi _{j}}{\pi _{i}\pi _{j}}%
\right) \frac{I_{i}I_{j}}{\pi _{ij}}.  \label{DEF:Vhat}
\end{equation}
\end{theorem}

Details of the proofs of Theorems \ref{THM:theta
consistent}-\ref{THM:normality} are given in Appendix B.

\vskip 0.1in \noindent \textbf{Remark 3.2:} In reference
\citep{REF:BCO}, the number of knots is fixed, thus the bias
caused by spline approximation in developing the asymptotic theory
is ignored. It has been shown in many contexts of function
estimation that, by letting the number of knots increase with the
sample size at an appropriate rate, spline estimate of an unknown
function can achieve the optimal nonparametric rate of
convergence; see \citep{REF:H,REF:WY}. For this purpose, in this
article, $n_{N}^{1/6}\ll J_{N}\ll n_{N}^{1/5}\left\{ \log \left(
n_{N}\right) \right\} ^{-2/5}$, as shown in Assumption (A4).

\vskip 0.1in \noindent \textbf{Remark 3.3:} As one one referee
pointed out, the asymptotics with the number of knots allowed to
grow is much more challenging, and only very recent work tackles
this problem, e.g. \citep{REF:CKO,REF:LR}. However, the results
obtained in this article are not directly comparable to those
obtained in \citep{REF:CKO,REF:LR} due to different settings of
the model. The problem in \citep{REF:CKO,REF:LR} is a purely
nonparametric curve estimation problem and the objective is to
study the asymptotics of the curve estimators fitted with
penalized splines. While the problem here is a semi-parametric one
and the main interest is in estimating the parametric component
$\mathbf{\theta}$. At the population level, it has been shown that
$\mathbf{\theta}_{0}$ should be estimable at the usual root-$n$
rate of convergence using similar techniques as deriving the
asymptotics of maximum likelihood estimators. In this article
examination of the approximation results of the derivatives (up to
the $2$nd order) of the risk function in (\ref{DEF:Rtheta}) by
their empirical versions implies that a range of smoothing
parameter is allowed for the desired asymptotics; see Appendix A.
This differs from nonparametric curve estimation in
\citep{REF:CKO,REF:LR} in which the optimal choice of the
smoothing parameter is required to achieve the optimal rate of
convergence.

\section{Algorithm}\label{SEC:algorithm}
\renewcommand{\thelemma}{{\sc \thesection.\arabic{lemma}}}
\renewcommand{\theequation}{\thesection.\arabic{equation}}
\setcounter{equation}{0} \setcounter{lemma}{0}

In this section, the actual procedure is described to implement
the estimation of $\mathbf{\theta }_{0}$ and $t_{y}$. I first
introduce some new notation. For any fixed $\mathbf{\theta }$, write $\mathbf{P}_{\mathbf{%
\theta },s}=\mathbf{B}_{\mathbf{\theta },s}\left( \mathbf{B}_{\mathbf{\theta }%
,s}^{T}\mathbf{W}_{s}\mathbf{B}_{\mathbf{\theta},s}\right)^{-1}\mathbf{B}_{\mathbf{\theta}%
,s}^{T}\mathbf{W}_{s}$ as the sample projection matrix onto the
cubic spline space.
For any $q=1,...,d$, write $\mathbf{%
\dot{B}}_{q}=\frac{\partial }{\partial \theta _{q}}\mathbf{B}_{\mathbf{%
\theta }}$, $\mathbf{\dot{P}}_{q}=\frac{\partial }{\partial \theta _{q}}%
\mathbf{P}_{\mathbf{\theta }}$ as the first order derivatives of $%
\mathbf{B}_{\mathbf{\theta }}$ and
$\mathbf{P}_{\mathbf{\theta},s}$ with respect to $\mathbf{\theta
}$. Write $\mathbf{\theta }_{-d}=\left( \theta _{1},...,\theta
_{d-1}\right) ^{T}$. Let $\hat{S}^{*}(\mathbf{\theta }_{-d})$ be the score vector of the risk function $\hat{R}%
^{*}\left( \mathbf{\theta }_{-d}\right)=\hat{R}\left( \theta
_{1},\theta _{2},...,\theta _{d-1},\sqrt{1-\left\| \mathbf{\theta
}_{-d}\right\| _{2}^{2}}\right)$,
that is, $\hat{S}^{*}(\mathbf{%
\theta }_{-d})=\frac{\partial }{\partial \mathbf{\theta }_{-d}}\hat{R}%
^{*}\left( \mathbf{\theta}_{-d}\right)$. The next lemma provides
the exact form of $\hat{S}^{*}(\mathbf{\theta }_{-d})$.

\begin{lemma}
\label{LEM:Shatstarmatrices} For $\hat{S}^{*}(\mathbf{\theta
}_{-d})$, the score vector of $\hat{R}^{*}\left( \mathbf{\theta
}_{-d}\right) $, one has
\begin{equation}
\hat{S}^{*}\left( \mathbf{\theta}_{-d}\right) =-n^{-1}\left\{\mathbf{y}_{s}^{T}%
\mathbf{\dot{P}}_{q}\mathbf{y}_{s}-\theta _{q}\mathbf{\theta }_{d}^{-1}\mathbf{y}_{s}%
^{T}\mathbf{\dot{P}}_{d}\mathbf{y}_{s}\right\} _{q=1}^{d-1},
\label{EQ:Shatstarmatrix}
\end{equation}
where for any $q=1,...,d$, $\mathbf{y}_{s}^{T}\mathbf{\dot{P}}_{q}%
\mathbf{y}_{s}=2\mathbf{y}_{s}^{T}\left( \mathbf{I}-\mathbf{P}_{\mathbf{\theta }%
,s}\right) \mathbf{\dot{B}}_{q}\left(\mathbf{B}_{\mathbf{\theta },s}^{T}\mathbf{W}_{s}\mathbf{%
B}_{\mathbf{\theta},s}\right) ^{-1}\mathbf{B}_{\mathbf{\theta
},s}^{T}\mathbf{W}_{s}\mathbf{y}_{s} $, and
$\mathbf{\dot{B}}_{q}=\left\{ J_{N}\left\{
B_{j,3}\left(z_{\mathbf{\theta
},i}\right)-B_{j+1,3}\left(z_{\mathbf{\theta },i}\right) \right\} \dot{F}%
_{d}\left( \mathbf{x}_{\mathbf{\theta },i}\right) x_{i,q}\right\}
_{i\in s,j=-3,...,J}$, with
\[
\dot{F}_{d}\left( x\right) =\frac{d}{dx}F_{d}=\frac{\Gamma \left(
d+1\right)
}{a\Gamma \left\{ \left( d+1\right) /2\right\} ^{2}2^{d}}\left( 1-\frac{x^{2}%
}{a^{2}}\right) ^{\frac{d-1}{2}}I\left( \left| x\right| \leq
a\right) .
\]
\end{lemma}

In practice, the estimation is implemented via the following
procedure.

Step 1. \textit{Standardize the auxiliary variables $\left\{ \mathbf{x}%
_{i}\right\} _{i\in U_{N}}$ and find the radius $a$ used in the
CDF transformation (\ref{DEF:Fd}) by calculating the
$100(1-\alpha)$ percentile of $\left\{ \Vert \mathbf{x}_{i}\Vert
\right\} _{i\in U_{N}}$ ($\alpha=0.01,0.05$ for example)}.

Step 2. \textit{Find the estimator $\hat{\mathbf{\theta }}$ of $\mathbf{%
\theta }_{0}$ by minimizing $\hat{R}$ in (\ref{DEF:Rhat}) through
the port optimization routine in the technical report of
\citep{REF:Gay}, with $\left( 0,0,...,1\right) ^{T}$ as the
initial value and the gradient vector $\hat{S}^{\ast }$ in
equation (\ref{EQ:Shatstarmatrix}). If $d<n$, one can take the
simple OLS estimator (after standardization) for
$\left\{y_{i},\mathbf{x}_{i}\right\}_{i\in s}$ with its last
coordinate positive.}

Step 3. \textit{Obtain the estimator $\hat{m}_{i}$ of
$m\left(\mathbf{x}_i\right)$, $i\in U_{N}$, by applying formula
(\ref{DEF:mihat}).}

Step 4. \textit{Calculate the sample design-based spline estimator
of $t_{y}$ in (\ref{DEF:tydiffhat})}.

\vskip 0.1in \noindent \textbf{Remark 4.1.} In Step 2, the number
of interior knots is
\begin{equation}
J=\min\left\{c_{1}\left[n^{1/5.5}\right], c_{2}\right\},
\label{EQ:numberofknots}
\end{equation}
where $c_{1}$ and $c_{2}$ are positive integers and $[\nu]$
denotes the integer part of $\nu$. The choice of the tuning
parameter $c_1$ makes little difference for a large sample, and
according to our asymptotic theory, there is no optimal way to set
these $c_{1}$ and $c_{2}$. I recommend using $c_{1}=1$ to save
computing for massive data sets and $c_{2}=5,...,10$ for smooth
monotonic or smooth unimodal regression as suggested by
\citep{REF:YR}.

\section{Empirical Results} \label{SEC:empirical results}
\renewcommand{\theequation}{\thesection.\arabic{equation}} %
\renewcommand{\thesubsection}{\thesection.\arabic{subsection}}
\setcounter{subsection}{0}\setcounter{corollary}{0}\setcounter{equation}{0}

In this section, empirical results are provided to demonstrate the
applicability of the proposed methodology. A computing package in
R can be downloaded from the following website:
http://lilywang.myweb.uga.edu/research.htm. Besides the spline
single-index (SIM) estimators proposed in the article, I have
obtained for comparison the performance of three other estimators:
Horvitz-Thompson estimator (HT) in equation (\ref{DEF:HT}), linear
regression estimator (LREG) without interaction terms in Chapter 6
of \citep{REF:SSW}, and spline additive estimator (AM) in
\citep{REF:BCO} with degrees $1,2$ and $3$ and adaptive knots. The
number of knots $J_{N}$ for the spline SIM estimator is selected
according to (\ref {EQ:numberofknots}).

\subsection{Simulated Population}

To illustrate the finite-sample behavior of the estimator $\hat{t}_{y,%
\text{diff}}$, some simulation results are presented. For the
superpopulation model (\ref{model}), the following six mean
functions are considered:
\[
\begin{array}{ll}
2\text{-dimension (Linear): } & m_{1}\left( \mathbf{x}\right) =x_{1}+x_{2} \\
2\text{-dimension (Quadratic): } & m_{2}\left( \mathbf{x}\right) =1+\left( x_{1}+x_{2}\right) ^{2} \\
2\text{-dimension (Bump 1): } & m_{3}\left( \mathbf{x}\right)
=x_{1}+x_{2}+4\exp \left\{ -\left(
x_{1}+x_{2}\right) ^{2}\right\}  \\
2\text{-dimension (Bump 2): } & m_{4}\left( \mathbf{x}\right)
=x_{1}+x_{2}+4\exp \left\{ -\left(
x_{1}+x_{2}\right) ^{2}\right\} +\sqrt{x_{1}^{2}+x_{2}^{2}} \\
4\text{-dimension (Sinusoid): } & m_{5}\left( \mathbf{x}\right) =\sin (\pi \mathbf{x}^{T}\mathbf{\theta }%
_{0}),\mathbf{\theta }_{0}=\left. \left( 1,1,0,1\right)
^{T}\right/ \sqrt{3}
\\
10\text{-dimension (Sinusoid): } & m_{6}\left( \mathbf{x}\right) =\sin (\pi \mathbf{x}^{T}\mathbf{\theta }%
_{0}),\mathbf{\theta }_{0}=\left. \left( 1,1,0,...,0,1\right)
^{T}\right/ \sqrt{3}
\end{array}
\]
These represent various correct and incorrect single-index model
specifications. Function $m_{1}$ is a simple linear additive
function with two auxiliary variables, and it is also a linear
single-index function; Functions $m_{2}$, $m_{3}$, $m_{5}$ and
$m_{6}$ are some very common single-index models, but unlike
$m_{1}$, they are not additive so that the purely linear or
additive model would be misspecified. Function $m_{4}$ is neither
a genuine single-index nor a genuine additive function so that any
of the above models would be misspecified. However, because the
single-index model in this article is identified by the best
approximation (see equation (\ref{model}))
to the multivariate mean function, the estimator $\hat{t}_{y,%
\text{diff}}$ is expected to be robust in this case.

The auxiliary vector $\left\{\mathbf{x}_{i}\right\}_{i\in U_{N}}$
is generated from i.i.d uniform $(0,1)$ random vectors. The
population values $y_{i}$'s are generated from the mean functions
by adding i.i.d $N\left( 0,\sigma ^{2}\right) $ errors with
$\sigma =0.1$ and $0.4$. The population is of size $N=1000$.
Samples are generated by simple random sampling using sample size $%
n_{N}=50, 100$ and $200$. For each combination of mean function, standard deviation and sample size, $%
1000$ replicates are selected from the same population, the
estimators are calculated, and the design bias, design variance
and the design mean squared errors are estimated.

Table \ref{TAB:AMSE-theta} lists the average mean squared errors
(AMSE) of the spline estimators $\hat{\mathbf{\theta}}$ in
(\ref{DEF:thetahat}) based on $d$ dimensions
\begin{equation}
\text{AMSE}\left(\hat{\mathbf{\theta}}\right)= \frac{1}{d}\sum_{q=1}^{d}%
\text{MSE}\left(\hat{\mathbf{\theta}}_{q}\right),
\label{DEF:AMSE-theta}
\end{equation}
from which one sees that, even for small sample size, the
estimators $\hat{\mathbf{\theta}}$ are very accurate for all the
population models, and the precision is improved when sample size
$n_{N}$ increases.

In terms of the design biases, the percent relative design biases
\[
\left\{E_{p}[\hat{t}_{y,diff}]-t_{y}\right\}/{t_{y}}\times 100\%
\]
defined in \citep{REF:BO} have been measured for all the above
models. It is found that the relative design biases of the SIM
estimators are quite small (less than one percent for all cases in
the simulation) even for sample size $n_{N}=50$.

\setlength{\unitlength}{1cm}
\begin{table}
\small \caption{AMSE of the spline estimators
$\hat{\mathbf{\theta}}$ defined in (\ref{DEF:AMSE-theta}).$^{\rm
a}$}
{\begin{tabular*}{1.0\textwidth}{@{\extracolsep{\fill}}cccccccc}
\hline \hline%
$\sigma$ & $n_{N}$ & 1 & 2 & 3 & 4 & 5 & 6 \\ \hline
\multirow{3}{*}{$0.1$} & $50$ & $0.00076$ & $0.0005$ & $0.00027$ & $0.00157$ & $0.00291$ & $0.00482$ \\
& $100$ & $0.0004$ & $0.00027$ & $0.00014$ & $0.00071$ & $0.0013$ & $0.00226$ \\
& $200$ & $0.00019$ & $0.00016$ & $0.00008$ & $0.00038$ & $0.0007$ & $0.00124$ \\
\hline%
\multirow{3}{*}{$0.4$} & $50$ & $0.01732$ & $0.00326$ & $0.00504$ & $0.01981$ & $0.00958$ & $0.01427$ \\
& $100$ & $0.00819$ & $0.00177$ & $0.00257$ & $0.00879$ & $0.00453$ & $0.00696$ \\
& $200$ & $0.00423$ & $0.00089$ & $0.00129$ & $0.00398$ & $0.00233$ & $0.00372$ \\
\hline%
\end{tabular*}}
\footnotesize{$^{\rm a}$Based on $1000$ replications of $n_{N}$
simple random samples from population of size
$N=1000$.}\label{TAB:AMSE-theta}
\end{table}

Table \ref{TAB:ratio} shows the ratios of design mean squared
errors (MSE) for HT, LREG and AM estimators to the MSE for the
proposed spline SIM estimator. From this table, one sees that the
model-assisted estimators, LREG, AM and SIM estimators, perform
much better than the simple HT estimators regardless the type of
mean function and standard error. For $m_{1}$, LREG is expected to
be the preferred estimator, since the assumed model is correctly
specified. The AM and SIM estimators have similar behavior in this
case, and the MSE ratios of AM to SIM are close to $1$. However,
not much efficiency lost by using SIM and AM instead of LREG. The
MSE ratios of LREG to SIM are at least 0.78 for all cases. For the
rest of the population, the SIM estimators perform consistently
better than LREG and AM estimators because the interactions
between the auxiliary variables have been completely ignored for
LREG and AM estimators. For $m_{4}$, it is not a genuine
single-index function, but SIM estimators are still much more
accurate than HT, LREG and AM estimators, confirmative to the
theory that the proposed estimators are robust against the
deviation from single-index model.

To see how fast the computation is, Table \ref{TAB:ratio} provides
the average time (based on $1000$ replications) of obtaining the
SIM estimators on an ordinary PC with Intel Pentium IV 1.86 GHz
processor and 1.0 GB RAM. It shows that the proposed SIM
estimation is extremely fast. For instance, for Model $6$, the SIM
estimation of a $10$-dimensional sample of size $200$ takes on
average $0.23$ second. I have also carried out the simulation with
sample size $n_{N}=5000$ generated from the population of size
$50000$. Remarkably, it takes on average less than $8$ seconds to
get the SIM estimators for all the above models.

\setlength{\unitlength}{1cm}
\begin{table}
\small \caption{Ratio of MSE of the HT, LREG and additive
model-assisted estimators (AM) to the single-index model-assisted
estimators (SIM) and the average computing time of the SIM. $^{\rm
a}$}
{\begin{tabular*}{1.0\textwidth}{@{\extracolsep{\fill}}ccc|rrrrr|c}
\hline \hline \multirow{3}{*}{\textbf{Model}} &
\multirow{3}{*}{$\sigma$} &
\multirow{3}{*}{$n_{N}$} & \multicolumn{5}{c|}{\textbf{MSE Ratio}} & {\textbf{Time of SIM}} \\
\cline{4-8}
& & & \multirow{2}{*}{\textbf{HT}} & \multirow{2}{*}{\textbf{LREG}} & \multicolumn{3}{c|}{\textbf{AM}} & {\textrm{%
(seconds)}} \\ 
& & & & & degree$=1$ & degree$=2$ & degree$=3$ & \\%
\hline%
\multirow{6}{*}{$1$} &
\multirow{3}{*}{$0.1$} & $50$  & $12.40$ & $0.78$ & $0.97$ & $1.13$ & $3.05$ & $0.11$\\
&  & $100$ & $14.10$ & $0.88$ & $0.97$ & $0.96$ & $1.02$ & $0.12$\\
&  & $200$ & $14.29$ & $0.90$ & $0.92$ & $0.92$ & $0.95$ & $0.17$\\
& \multirow{3}{*}{$0.4$} & $50$ & $1.15$ & $0.82$ & $1.02$ & $1.19$ & $3.93$ & $0.14$\\
&  & $100$ & $1.74$ & $0.92$ & $1.01$ & $1.01$ & $1.07$ & $0.13$\\
&  & $200$ & $1.72$ & $0.95$ & $0.98$ & $0.97$ & $1.01$ & $0.18$\\
\hline%
\multirow{6}{*}{$2$} & \multirow{3}{*}{$0.1$} & $50$ &
$35.32$ & $2.75$ & $2.63$ & $3.15$ & $7.32$ & $0.11$\\
&  & $100$ & $41.46$ & $3.48$ & $2.67$ & $2.84$ & $3.19$ & $0.12$\\
&  & $200$ & $43.78$ & $4.44$ & $2.75$ & $2.75$ & $2.82$ & $0.18$\\
& \multirow{3}{*}{$0.4$} & $50$  & $4.03$ & $1.03$ & $1.17$ & $1.36$ & $4.11$ & $0.11$\\
&  & $100$ & $4.62$ & $1.20$ & $1.18$ & $1.20$ & $1.31$ & $0.12$\\
&  & $200$ & $4.65$ & $1.35$ & $1.17$ & $1.16$ & $1.20$ & $0.18$\\
\hline%
\multirow{6}{*}{$3$} &
\multirow{3}{*}{$0.1$} & $50$ & $34.47$ & $4.24$ & $4.72$ & $5.75$ & $15.53$ & $0.16$\\
&  & $100$ & $39.23$ & $5.38$ & $4.87$ & $5.14$ & $5.56$ & $0.19$\\
&  & $200$ & $42.20$ & $6.82$ & $5.22$ & $5.30$ & $5.36$ & $0.32$\\
& \multirow{3}{*}{$0.4$} & $50$ & $2.98$ & $1.09$ & $1.30$ & $1.53$ & $5.28$ & $0.16$\\
&  & $100$ & $3.36$ & $1.23$ & $1.26$ & $1.28$ & $1.38$ & $0.19$\\
&  & $200$ & $3.80$ & $1.42$ & $1.30$ & $1.30$ & $1.33$ & $0.30$\\
\hline%
\multirow{6}{*}{$4$} &
\multirow{3}{*}{$0.1$} & $50$ & $10.13$ & $2.88$ & $2.68$ & $3.24$ & $8.17$ & $0.16$\\
&  & $100$ & $11.01$ & $3.50$ & $2.57$ & $2.68$ & $2.86$ & $0.17$\\
&  & $200$ & $12.67$ & $4.63$ & $2.83$ & $2.86$ & $2.88$ & $0.27$\\
& \multirow{3}{*}{$0.4$} & $50$ & $1.59$ & $1.03$ & $1.19$ & $1.40$ & $4.72$ & $0.17$\\
&  & $100$ & $1.80$ & $1.16$ & $1.14$ & $1.15$ & $1.22$ & $0.19$\\
&  & $200$ & $2.10$ & $1.34$ & $1.16$ & $1.16$ & $1.19$ & $0.30$\\
\hline%
\multirow{6}{*}{$5$} &
\multirow{3}{*}{$0.1$} & $50$ & $18.73$ & $3.51$ & $5.56$ & $8.83$ & $10.79$ & $0.15$\\
&  & $100$ & $25.16$ & $4.42$ & $4.82$ & $5.43$ & $6.67$ & $0.14$\\
&  & $200$ & $29.41$ & $4.97$ & $4.64$ & $4.82$ & $5.02$ & $0.21$\\
& \multirow{3}{*}{$0.4$} & $50$ & $2.54$ & $1.11$ & $1.97$ & $3.43$ & $17.4$ & $0.12$\\
&  & $100$ & $3.08$ & $1.28$ & $1.53$ & $1.63$ & $1.95$ & $0.14$\\
&  & $200$ & $3.43$ & $1.30$ & $1.39$ & $1.41$ & $1.47$ & $0.21$\\
\hline%
\multirow{6}{*}{$6$} &
\multirow{3}{*}{$0.1$} & $50$ & $8.33$ & $1.63$ & $9.86$ & $7.20$ & $5.26$ & $0.15$\\
&  & $100$ & $13.39$ & $2.21$ & $4.41$ & $5.73$ & $10.53$ & $0.15$\\
&  & $200$ & $18.55$ & $2.90$ & $3.22$ & $3.56$ & $4.09$ & $0.23$\\
& \multirow{3}{*}{$0.4$} & $50$ & $2.08$ & $0.99$ & $6.16$ & $4.59$ & $3.11$ & $0.17$\\
&  & $100$ & $2.63$ & $1.09$ & $2.15$ & $3.01$ & $5.06$ & $0.16$\\
&  & $200$ & $3.20$ & $1.21$ & $1.44$ & $1.59$ & $1.82$ & $0.23$\\
\hline%
\end{tabular*}}
\footnotesize{$^{\rm a}$ Based on $1000$ replications of simple
random sampling from population of size
$N=1000$.}\label{TAB:ratio}
\end{table}

\subsection{MU281 data}

The MU284 data set from Appendix B of \citep{REF:SSW} contains
data about Swedish municipalities. The study variable $y$ is
RMT$85\times 10^{-3}$, where RMT85 is municipal tax receipts in
1985. Two auxiliary variables $x_{1}$ (CS82) and $x_{2}$ (SS82)
are used, where $x_{1}$ is the number of Conservative Party seats
in the municipal council, and $x_{2}$ is the number of Social
Democrat Party seats. The largest three cities according to the
variable population in 1975 (pop75) are discarded because they are
huge outliers and would be treated separately in practice. The
population total of $N=281$ Swedish Municipalities, $t_{y}$, is
found to be $53.1510$. The ``oracle'' estimator
$\tilde{\mathbf{\theta}}$ (see (\ref{DEF:thetatilde})) at the
population level is found to be $\left(0.8412,0.5406\right)^{T}$.

A Monte Carlo simulation is carried out in which $1000$ repeated
SRS samples (each with $n=50$ and $100$) are drawn from the MU281
population of Swedish municipalities. To demonstrate the closeness
of the spline estimator $\hat{\mathbf{\theta }}$ to the ``oracle''
index parameter $\tilde{\mathbf{\theta }}$, Table
\ref{TAB:thetahat(MU281)} lists the sample mean (MEAN), design
bias (BIAS), design standard deviation (SD), the design mean
squared error (MSE) and the AMSE in (\ref{DEF:AMSE-theta}) of
$\hat{\mathbf{\theta }}$. From this table, one sees that the
sample-based estimators $\hat{\mathbf{\theta}}$ are very accurate
even for sample size $50$. As what is expected, when the sample
size increases, the coefficient is more accurately estimated.
Table \ref{TAB:tyhat(MU281)} shows the performance of the HT,
LREG, AM and SIM estimators of $t_{y}$. One sees from this table
that the model-assisted estimators are much more accurate than the
simple HT estimators. Among all the model-assisted estimators, the
spline SIM estimators are better than other estimators in terms of
the MSE.

\setlength{\unitlength}{1cm}
\begin{table}
\small
\caption{Spline estimators $\hat{\mathbf{\theta}}$ on MU281 Data
$^{\rm a}$}
{\begin{tabular*}{1.0\textwidth}{@{\extracolsep{\fill}}cc|rrrrr}
\hline\hline%
$n_{N}$ & $\mathbf{\theta}$ & MEAN & BIAS & SD & MSE & AMSE \\
\hline%
\multirow{2}{*}{$50$} & $\mathbf{\theta}_{1}$ & $0.8343$ & $-0.0069$ & $0.0643$ & $0.0468$ & \multirow{2}{*}{0.0507}  \\
& $\mathbf{\theta}_{2}$ & $0.5395$ & $-0.0013$ & $0.0940$ &
$0.0546$ \\ \hline%
\multirow{2}{*}{$100$} & $\mathbf{\theta}_{1}$ & $0.8412$ & $0.0001$ & $0.0359$ & $0.0465$ & \multirow{2}{*}{0.0480} \\
& $\mathbf{\theta}_{2}$ & $0.5365$ & $-0.0042$ & $0.0563$ & $0.0495$ & \\
\hline%
\end{tabular*}}
\footnotesize{$^{\rm a}$ Based on $1000$ replications of simple
random sampling from population of $N=281$ Swedish
Municipalities.}\label{TAB:thetahat(MU281)}
\end{table}

Table \ref{TAB:tyhat(MU281)} shows the performance of the HT,
LREG, AM and SIM estimators of $t_{y}$. One sees from this table
that the model-assisted estimators are much more accurate than the
simple HT estimators. Among all the model-assisted estimators, the
spline SIM estimators are better than other estimators in terms of
the MSE.

\setlength{\unitlength}{1cm}
\begin{table}
\small
\caption{Estimators of $t_{y}$ on MU281 Data $^{\rm a}$}
{\begin{tabular*}{1.0\textwidth}{@{\extracolsep{\fill}}cl|rrrr}
\hline\hline%
$n_{N}$ & Estimator & MEAN & BIAS & SD & MSE \\ \hline%
\multirow{5}{*}{$50$} & HT & $53.0522$ & $-0.0988$ & $7.1936$ & $51.7051$ \\
& LREG & $52.7476$ & $-0.4034$ & $3.6387$ & $13.3893$ \\
& AM (degree $=1$) & $52.7784$ & $ -0.3726$ & $3.8640$ & $15.0539$ \\
& AM (degree $=2$) & $53.1789$ & $0.0279$ & $3.8243$ & $14.6113$ \\
& AM (degree $=3$) & $53.2588$ & $ 0.1078$ & $4.2464$ & $18.0254$ \\
& SIM & $52.9247$ & $-0.2264$ & $3.4645$ & $12.0416$ \\ \hline%
\multirow{5}{*}{$100$} & HT & $53.2739$ & $ 0.1229$ & $4.5926$ & $21.0859$ \\
& LREG & $52.9646$ & $-0.1864$ & $2.3892$ & $5.7373$ \\
& AM (degree $=1$) & $53.0125$ & $-0.1385$ & $2.4677$ & $6.1026$ \\
& AM (degree $=2$) & $53.1217$ & $-0.0293$ & $2.3538$ & $5.5359$ \\
& AM (degree $=3$) & $53.1197$ & $-0.0313$ & $2.3936$ & $5.7248$ \\
& SIM & $53.0791$ & $-0.0719$ & $2.3377$ & $5.4646$ \\
\hline%
\end{tabular*}}
\footnotesize{$^{\rm a}$ Based on $1000$ replications of simple
random sampling from population of $N=281$ Swedish
Municipalities.}\label{TAB:tyhat(MU281)}
\end{table}

\section*{Appendix A. Preliminaries}
\renewcommand{\thetheorem}{{\sc A.\arabic{theorem}}} \renewcommand{%
\theproposition}{{\sc A.\arabic{proposition}}}
\renewcommand{\thelemma}{{\sc
A.\arabic{lemma}}} \renewcommand{\thecorollary}{{\sc
A.\arabic{corollary}}} \renewcommand{\theequation}{A.\arabic{equation}} %
\renewcommand{\thesubsection}{{A.\arabic{subsection}}} %
\setcounter{equation}{0} \setcounter{lemma}{0} \setcounter{proposition}{0} %
\setcounter{theorem}{0}
\setcounter{subsection}{0}\setcounter{corollary}{0}

Let matrices
\begin{equation}
\mathbf{V}_{\mathbf{\theta }}=\frac{1}{N}\mathbf{B}_{%
\mathbf{\theta }}^{T}\mathbf{B}_{\mathbf{\theta }}, \rm{ } \mathbf{D}_{\mathbf{%
\theta }}=\frac{1}{N}\mathbf{B}_{\mathbf{\theta }}^{T}\mathbf{y}.
\label{DEF:Vtheta}
\end{equation}
The following lemma provides the uniform upper bound of $\left\| \mathbf{V}%
_{\mathbf{\theta }}^{-1}\right\| _{\infty }$.

\begin{lemma}
\label{LEM:ninvBB} Under Assumptions (A1)-(A4), there exist
constants $0<c_{V}<C_{V}$ such that with $\xi$-probability 1
\[
c_{V}J_{N}^{-1}\left\| \mathbf{w}\right\| _{2}^{2}\mathbf{\leq w}^{T}\mathbf{V}%
_{\mathbf{\theta }}\mathbf{w}\leq C_{V}J_{N}^{-1}\left\|
\mathbf{w}\right\| _{2}^{2}.
\]
Consequently, there exists a constant $C>0$ such that with
$\xi$-probability 1
\begin{equation}
\sup\limits_{\mathbf{\theta }\in S_{c}^{d-1}}\left\| \mathbf{V}_{\mathbf{%
\theta }}^{-1}\right\| _{\infty }\leq CJ_{N}. \label{EQ:ninvBB}
\end{equation}
\end{lemma}
The result follows directly from Theorem 5.4.2 of \citep{REF:DL}
and Assumptions (A1)-(A4), thus omitted.

In the following, let $S_{c}^{d-1}$ be a cap shape subset of
$S_{+}^{d-1}$,
\[
S_{c}^{d-1}=\left\{ \left( \theta _{1},...,\theta _{d}\right)
|\sum_{q=1}^{d}\theta _{q}^{2}=1,\theta _{d}\geq c\right\} ,c\in
\left( 0,1\right) .
\]
Clearly, for an appropriate choice of $c$, $\mathbf{\theta
}_{0}\in S_{c}^{d-1}$, which I assume in the rest of the article.

\begin{lemma}
\label{LEM:phihat-phitilde} Under Assumptions (A1)-(A5), for $k=0,1,2$%
\begin{equation}
\sup\limits_{\mathbf{\theta }\in S_{c}^{d-1}}\sup_{z\in \left[
0,1\right]
}\left| \frac{\partial ^{k}}{\partial \mathbf{\theta }^{k}}\hat{\varphi}_{%
\mathbf{\theta }}\left( z\right) -\frac{\partial ^{k}}{\partial \mathbf{%
\theta }^{k}}\tilde{\varphi}_{\mathbf{\theta }}\left( z\right)
\right| =O_{p}\left( J_{N}^{k}/n_{N}^{1/2}\right) ,
\label{EQ:phihat-tilde-derivative}
\end{equation}
where $\tilde{\varphi}_{\mathbf{\theta }}$ and $\hat{\varphi}_{\mathbf{%
\theta }}$ are given in (\ref{DEF:mthetatilde}) and
(\ref{DEF:mthetahat}).
\end{lemma}

\begin{proof}
First we show the case when $k=0$. Let
\[
\mathbf{V}_{\mathbf{\theta },\pi }=\frac{1}{N}\mathbf{B}_{\mathbf{\theta }%
,s}^{T}\mathbf{W}_{s}\mathbf{B}_{\mathbf{\theta },s},\mathbf{D}_{\mathbf{%
\theta },\pi }=\frac{1}{N}\mathbf{B}_{\mathbf{\theta },s}^{T}\mathbf{W}_{s}%
\mathbf{y}_{s},
\]
be the sample version of matrices $\mathbf{V}_{\mathbf{\theta }}$ and $%
\mathbf{D}_{\mathbf{\theta }}$ in (\ref{DEF:Vtheta}), then
\begin{equation}
\hat{\varphi}_{\mathbf{\theta }}\left( z\right) -\tilde{\varphi}_{\mathbf{%
\theta }}\left( z\right) =\mathbf{B}^{T}\left( z\right) \mathbf{V}_{\mathbf{%
\theta },\pi }^{-1}\mathbf{D}_{\mathbf{\theta },\pi
}-\mathbf{B}^{T}\left(
z\right) \mathbf{V}_{\mathbf{\theta }}^{-1}\mathbf{D}_{\mathbf{\theta }} \\
\equiv \zeta \left( \mathbf{V}_{\mathbf{\theta },\pi },\mathbf{D}_{\mathbf{%
\theta ,}\pi }\right) ,  \label{DEF:zeta}
\end{equation}
which is a nonlinear function of the following $\pi $ estimators:
\[
v_{\mathbf{\theta ,}\pi ,jj^{\prime }}=\frac{1}{N}\sum_{i\in
s}\left.
B_{j,4}\left( z_{\mathbf{\theta }i}\right) B_{j^{\prime },4}\left( z_{%
\mathbf{\theta }i}\right) \right/ \pi _{i}, d_{\mathbf{\theta ,}\pi ,j}=%
\frac{1}{N}\sum_{i\in s}\left. B_{j,4}\left( z_{\mathbf{\theta
}i}\right) y_{i}\right/ \pi _{i},
\]
the components of $\mathbf{V}_{\mathbf{\theta },\pi }$ and $\mathbf{D}_{%
\mathbf{\theta },\pi }$, respectively. Thus the fist order derivative of $%
\zeta $ in (\ref{DEF:zeta}) with respect to $v_{\mathbf{\theta
,}\pi ,jj^{\prime }}$ and $d_{\mathbf{\theta ,}\pi ,j}$ can be
written as
\begin{eqnarray*}
\frac{\partial \zeta }{\partial v_{\mathbf{\theta ,}\pi ,jj^{\prime }}} &=&%
\mathbf{B}^{T}\left( z\right) \left( -\mathbf{V}_{\mathbf{\theta },\pi }^{-1}%
\mathbf{\Lambda }_{jj^{\prime }}\mathbf{V}_{\mathbf{\theta },\pi
}^{-1}\right) \mathbf{D}_{\mathbf{\theta ,}\pi },-3\leq j\leq
j^{\prime
}\leq J, \\
\frac{\partial \zeta }{\partial d_{\mathbf{\theta ,}\pi ,j}} &=&\mathbf{B}%
^{T}\left( z\right) \mathbf{V}_{\mathbf{\theta },\pi }^{-1}\mathbf{\lambda }%
_{j},-3\leq j\leq J,
\end{eqnarray*}
where $\mathbf{\lambda}_{j}$ is a $\left( J+1\right) $-vector with
the $j$th component equal to one, zeros elsewhere; and
$\mathbf{\Lambda }_{jj^{\prime }}$ is a $\left( J+1\right) \times
\left( J+1\right) $ matrix with the value $1$ in positions $\left(
j,j^{\prime }\right) $ and $\left( j^{\prime },j\right) $ and the
value $0$ everywhere else.

Denote the components of $\mathbf{V}_{\mathbf{\theta }}$ and $\mathbf{D}_{%
\mathbf{\theta }}$ by
\[
v_{\mathbf{\theta },jj^{\prime }}=\frac{1}{N}\sum_{i\in
U_{N}}B_{j,4}\left(
z_{\mathbf{\theta }i}\right) B_{j^{\prime },4}\left( z_{\mathbf{\theta }%
i}\right) ,d_{\theta ,j}=\frac{1}{N}\sum_{i\in U_{N}}B_{j,4}\left( z_{%
\mathbf{\theta }i}\right) y_{i},
\]
respectively. Using the Taylor linearization, one can approximate
the function $\zeta $ in (\ref{DEF:zeta}) by a linear one, i.e.
\begin{eqnarray*}
&&\left. \zeta \left( \mathbf{V}_{\mathbf{\theta },\pi },\mathbf{D}_{\mathbf{%
\theta ,}\pi }\right) =\sum_{j=-3}^{J}\mathbf{B}^{T}\left( z\right) \mathbf{V%
}_{\mathbf{\theta }}^{-1}\mathbf{\lambda }_{j}\left(
d_{\mathbf{\theta ,}\pi
,j}-d_{\mathbf{\theta ,}j}\right) \right.  \\
&&-\sum_{-3\leq j\leq j^{\prime }\leq J}\mathbf{B}^{T}\left(
z\right) \left(
\mathbf{V}_{\mathbf{\theta }}^{-1}\mathbf{\Lambda }_{jj^{\prime }}\mathbf{V}%
_{\mathbf{\theta }}^{-1}\right) \mathbf{D}_{\mathbf{\theta }}\left( v_{%
\mathbf{\theta ,}\pi ,jj^{\prime }}-v_{\mathbf{\theta },jj^{\prime
}}\right) +R_{iN},
\end{eqnarray*}
where the remainder term
\begin{eqnarray*}
&&\left. R_{iN}=\hat{\varphi}_{\mathbf{\theta }}\left( z\right) -\tilde{%
\varphi}_{\mathbf{\theta }}\left( z\right) -\sum_{j=-3}^{J}\mathbf{B}%
^{T}\left( z\right) \mathbf{V}_{\mathbf{\theta }}^{-1}\mathbf{\lambda }%
_{j}\left( d_{\mathbf{\theta ,}\pi ,j}-d_{\mathbf{\theta
,}j}\right) \right.
\\
&&+\sum_{-3\leq j\leq j^{\prime }\leq J}\mathbf{B}^{T}\left(
z\right) \left(
\mathbf{V}_{\mathbf{\theta }}^{-1}\mathbf{\Lambda }_{jj^{\prime }}\mathbf{V}%
_{\mathbf{\theta }}^{-1}\right) \mathbf{D}_{\mathbf{\theta }}\left( v_{%
\mathbf{\theta ,}\pi ,jj^{\prime }}-v_{\mathbf{\theta },jj^{\prime
}}\right) .
\end{eqnarray*}
Note that
\begin{eqnarray*}
&&\mathbf{B}^{T}\left( z\right) \mathbf{V}_{\mathbf{\theta }}^{-1}\mathbf{%
\lambda }_{j}\left( d_{\mathbf{\theta ,}\pi ,j}-d_{\mathbf{\theta ,}%
j}\right) =\mathbf{B}^{T}\left( z\right) \mathbf{V}_{\mathbf{\theta }%
}^{-1}\sum_{j=-3}^{J}\mathbf{\lambda }_{j}\left( d_{\mathbf{\theta
,}\pi
,j}-d_{\mathbf{\theta ,}j}\right)  \\
&=&\frac{1}{N}\sum_{i\in U_{N}}\mathbf{B}^{T}\left( z\right) \mathbf{V}_{%
\mathbf{\theta }}^{-1}\mathbf{B}_{\mathbf{\theta
}}^{T}e_{i}y_{i}\left( \frac{I_{i}}{\pi _{i}}-1\right) ,
\end{eqnarray*}
and
\begin{eqnarray*}
&&\sum_{-3\leq j\leq j^{\prime }\leq J}\mathbf{B}^{T}\left(
z\right) \left(
\mathbf{V}_{\mathbf{\theta }}^{-1}\mathbf{\Lambda }_{jj^{\prime }}\mathbf{V}%
_{\mathbf{\theta }}^{-1}\right) \mathbf{D}_{\mathbf{\theta }}\left( v_{%
\mathbf{\theta ,}\pi ,jj^{\prime }}-v_{\mathbf{\theta },jj^{\prime
}}\right)
\\
&=&\frac{1}{N}\sum_{i\in U_{N}}\left[ \mathbf{B}^{T}\left( z\right) \mathbf{V%
}_{\mathbf{\theta }}^{-1}\left\{ \sum_{-3\leq j\leq j^{\prime }\leq J}%
\mathbf{\Lambda }_{jj^{\prime }}B_{j,4}\left( z_{\mathbf{\theta
}i}\right)
B_{j^{\prime },4}\left( z_{\mathbf{\theta }i}\right) \right\} \mathbf{V}_{%
\mathbf{\theta }}^{-1}\mathbf{D}_{\mathbf{\theta }}\right] \left( \frac{I_{i}%
}{\pi _{i}}-1\right)  \\
&=&\frac{1}{N}\sum_{i\in U_{N}}\left[ \mathbf{B}^{T}\left( z\right) \mathbf{V%
}_{\mathbf{\theta }}^{-1}\mathbf{B}_{\mathbf{\theta }}^{T}e_{i}e_{i}^{T}%
\mathbf{B}_{\mathbf{\theta }}\mathbf{V}_{\mathbf{\theta }}^{-1}\mathbf{D}_{%
\mathbf{\theta }}\right] \left( \frac{I_{i}}{\pi _{i}}-1\right) .
\end{eqnarray*}
Thus
\[
\sup\limits_{\mathbf{\theta }\in S_{c}^{d-1}}\sup_{z\in \left[
0,1\right]
}\left| \sum_{j=-3}^{J}\mathbf{B}^{T}\left( z\right) \mathbf{V}_{\mathbf{%
\theta }}^{-1}\mathbf{\lambda }_{j}\left( d_{\mathbf{\theta ,}\pi ,j}-d_{%
\mathbf{\theta ,}j}\right) \right| =O_{p}\left(
n_{N}^{-1/2}\right) ,
\]
\begin{eqnarray*}
&&\sup\limits_{\mathbf{\theta }\in S_{c}^{d-1}}\sup_{u\in \left[
0,1\right] }\left| \sum_{-3\leq j\leq j^{\prime }\leq
J}\mathbf{B}^{T}\left( z\right)
\left( \mathbf{V}_{\mathbf{\theta }}^{-1}\mathbf{\Lambda }_{jj^{\prime }}%
\mathbf{V}_{\mathbf{\theta }}^{-1}\right) \mathbf{D}_{\mathbf{\theta }%
}\left( v_{\mathbf{\theta ,}\pi ,jj^{\prime }}-v_{\mathbf{\theta }%
,jj^{\prime }}\right) \right| =O_{p}\left( n_{N}^{-1/2}\right). \\
\end{eqnarray*}
Similarly, one can show that
\[
\sup_{\mathbf{\theta }\in S_{c}^{d-1}}\sup_{z\in \left[ 0,1\right]
}\left| R_{iN}\right| =o_{p}\left( n_{N}^{-1/2}\right) .
\]
Thus, (\ref{EQ:phihat-tilde-derivative}) holds when $k=0$. Next
according to (\ref{EQ:dotBp}) in Lemma \ref{LEM:PthetadotPpnorm},
the corresponding order on the right hand side of
(\ref{EQ:phihat-tilde-derivative}) will increases by $h^{-k}$ when
one takes the $k$th order derivative of
$\mathbf{B}_{\mathbf{\theta}}$. Similar arguments as given above
yields the desired results for $k=1$ and $2$.
\end{proof}

\begin{lemma}
\label{LEM:Rtilde-Rhat} Under Assumptions (A1)-(A5), with $p$%
-probability 1, one has
\begin{equation}
\lim_{N\rightarrow \infty }\sup\limits_{\mathbf{\theta }\in
S_{c}^{d-1}}\left| \hat{R}\left( \mathbf{\theta }\right)
-\tilde{R}\left( \mathbf{\theta }\right) \right|
=0,\label{EQ:Rhat-Rtilde}
\end{equation}
and
\begin{equation}
\sup\limits_{\mathbf{\theta }\in
S_{c}^{d-1}}\left| \frac{\partial ^{k}}{\partial \mathbf{\theta }^{k}}\hat{R}%
\left( \mathbf{\theta }\right) -\frac{\partial ^{k}}{\partial
\mathbf{\theta }^{k}}\tilde{R}\left( \mathbf{\theta }\right)
\right| =O_{p}\left(J_{N}^{k}/n_{N}^{1/2}\right),
\label{EQ:Rhat-Rtilde-div}
\end{equation}
where $\tilde{R}\left( \mathbf{\theta }\right) $ and $\hat{R}\left( \mathbf{%
\theta }\right) $ are the population-based and sample-based
empirical risk functions of $\mathbf{\theta }$ defined in
(\ref{DEF:Rtilde}) and (\ref {DEF:Rhat}).
\end{lemma}
\begin{proof}
Let
\begin{eqnarray*}
A_{N1}&=&\sup\limits_{\mathbf{\theta }\in S_{c}^{d-1}}\left|
N^{-1}\sum_{i\in
U_{N}}\left[ \left\{ y_{i}-\tilde{\varphi}_{\mathbf{\theta }}\left( z_{%
\mathbf{\theta }i}\right) \right\} ^{2}\left( \frac{I_{i}}{\pi _{i}}%
-1\right) \right] \right|, \\
A_{N2}&=&\sup\limits_{\mathbf{\theta }\in S_{c}^{d-1}}\left|
2N^{-1}\sum_{i\in
U_{N}}\left\{ \tilde{\varphi}_{\mathbf{\theta }}\left( z_{\mathbf{\theta }%
i}\right) -\hat{\varphi}_{\mathbf{\theta }}\left( z_{\mathbf{\theta }%
i}\right) \right\} \left\{ y_{i}-\tilde{\varphi}_{\mathbf{\theta }}\left( z_{%
\mathbf{\theta }i}\right) \right\} \frac{I_{i}}{\pi _{i}}\right|,\\
A_{N3}&=&\sup\limits_{\mathbf{\theta }\in S_{c}^{d-1}}\left|
N^{-1}\sum_{i\in
U_{N}}\left\{ \tilde{\varphi}_{\mathbf{\theta }}\left( z_{\mathbf{\theta }%
i}\right) -\hat{\varphi}_{\mathbf{\theta }}\left( z_{\mathbf{\theta }%
i}\right) \right\} ^{2}\frac{I_{i}}{\pi _{i}}\right| .
\end{eqnarray*}
Noting that
\[
\hat{R}\left( \mathbf{\theta }\right) -\tilde{R}\left( \mathbf{\theta }%
\right) =N^{-1}\sum_{i\in U_{N}}\left[ \left\{ y_{i}-\hat{\varphi}_{\mathbf{%
\theta }}\text{ }\left( z_{\mathbf{\theta }i}\right) \right\} ^{2}\frac{I_{i}%
}{\pi _{i}}-\left\{ y_{i}-\tilde{\varphi}_{\mathbf{\theta }}\left( z_{%
\mathbf{\theta }i}\right) \right\} ^{2}\right],
\]
one has \[
\sup\limits_{\mathbf{\theta }\in S_{c}^{d-1}}\left| \hat{R}\left( \mathbf{%
\theta }\right) -R\left( \mathbf{\theta }\right) \right| \leq
A_{N1}+A_{N2}+A_{N3}.
\]
Similar arguments as in the proof of Theorem \ref{THM:ADU}
entitles that $A_{N1}$ converges to zero with $p$-probability 1 as
$N\rightarrow \infty $. For $A_{N3}$, using the similar arguments
as in Lemma \ref{LEM:phihat-phitilde}, one has with
$p$-probability 1
\begin{eqnarray*}
A_{N3} &=&\sup\limits_{\mathbf{\theta }\in S_{c}^{d-1}}\left|
N^{-1}\sum_{i\in U_{N}}\left\{ \tilde{\varphi}_{\mathbf{\theta }}\left( z_{%
\mathbf{\theta }i}\right) -\hat{\varphi}_{\mathbf{\theta }}\left( z_{\mathbf{%
\theta }i}\right) \right\} ^{2}\frac{I_{i}}{\pi _{i}}\right| \\
&\leq &\sup\limits_{\mathbf{\theta }\in S_{c}^{d-1}}\sup_{u\in
\left[
0,1\right] }\left| N^{-1}\sum_{i\in U_{N}}\left\{ \tilde{\varphi}_{\mathbf{%
\theta }}\left( u\right) -\hat{\varphi}_{\mathbf{\theta }}\left(
u\right) \right\} ^{2}\right| \frac{1}{\lambda }\rightarrow 0.
\end{eqnarray*}
In terms of $A_{N2}$, note that
\begin{eqnarray*}
A_{N2} &\leq &\sup\limits_{\mathbf{\theta }\in S_{c}^{d-1}}\left\{
\left|
2N^{-1}\sum_{i\in U_{N}}\left\{ \tilde{\varphi}_{\mathbf{\theta }}\left( z_{%
\mathbf{\theta }i}\right) -\hat{\varphi}_{\mathbf{\theta }}\left( z_{\mathbf{%
\theta }i}\right) \right\} ^{2}\frac{I_{i}}{\pi _{i}}\right|
\right\} ^{1/2}
\\
&&\times \sup\limits_{\mathbf{\theta }\in S_{c}^{d-1}}\left\{
\left|
2N^{-1}\sum_{i\in U_{N}}\left\{ y_{i}-\tilde{\varphi}_{\mathbf{\theta }%
}\left( z_{\mathbf{\theta }i}\right) \right\} ^{2}\frac{I_{i}}{\pi _{i}}%
\right| \right\} ^{1/2} \\
&\leq &2A_{N3}^{1/2}\sup\limits_{\mathbf{\theta }\in
S_{c}^{d-1}}\left\{
E_{p}\left| N^{-1}\sum_{i\in U_{N}}\left\{ y_{i}-\tilde{\varphi}_{\mathbf{%
\theta }}\left( z_{\mathbf{\theta }i}\right) \right\} ^{2}\right| \frac{1}{%
\lambda }\right\} ^{1/2}.
\end{eqnarray*}
The definition of $\tilde{\varphi}_{\mathbf{\theta }}$ in (\ref
{DEF:mthetatilde}) implies that
\[
\lim_{N\rightarrow \infty }N^{-1}\sum_{i\in U_{N}}\left\{ y_{i}-\tilde{%
\varphi}_{\mathbf{\theta }}\left( z_{\mathbf{\theta }i}\right)
\right\} ^{2}<\infty .
\]
Thus $A_{N2}$ converges to zero with $p$-probability 1 as
$N\rightarrow \infty $, and (\ref{EQ:Rhat-Rtilde}) is proved.

Next note that
\begin{eqnarray*}
&&\left. \frac{\partial }{\partial \mathbf{\theta }}\hat{R}\left( \mathbf{%
\theta }\right) -\frac{\partial }{\partial \mathbf{\theta
}}\tilde{R}\left(
\mathbf{\theta }\right) =N^{-1}\sum_{i\in U_{N}}\left\{ y_{i}-\hat{\varphi}_{%
\mathbf{\theta }}\left( z_{\mathbf{\theta }i}\right) \right\} \left\{ \frac{%
\partial }{\partial \mathbf{\theta }}\hat{\varphi}_{\mathbf{\theta }}\left(
z_{\mathbf{\theta }i}\right) \right\} \frac{I_{i}}{\pi _{i}}\right.  \\
&&-N^{-1}\sum_{i\in U_{N}}\left\{ y_{i}-\tilde{\varphi}_{\mathbf{\theta }%
}\left( z_{\mathbf{\theta }i}\right) \right\} \left\{ \frac{\partial }{%
\partial \mathbf{\theta }}\tilde{\varphi}_{\mathbf{\theta }}\left( z_{%
\mathbf{\theta }i}\right) \right\}  \\
&=&N^{-1}\sum_{i\in U_{N}}y_{i}\left\{ \frac{\partial }{\partial \mathbf{%
\theta }}\hat{\varphi}_{\mathbf{\theta }}\left( z_{\mathbf{\theta
}i}\right)
-\frac{\partial }{\partial \mathbf{\theta }}\tilde{\varphi}_{\mathbf{\theta }%
}\left( z_{\mathbf{\theta }i}\right) \right\} \frac{I_{i}}{\pi _{i}}%
+N^{-1}\sum_{i\in U_{N}}y_{i}\frac{\partial }{\partial \mathbf{\theta }}%
\tilde{\varphi}_{\mathbf{\theta }}\left( z_{\mathbf{\theta
}i}\right) \left(
\frac{I_{i}}{\pi _{i}}-1\right)  \\
&&-N^{-1}\sum_{i\in U_{N}}\left\{ \hat{\varphi}_{\mathbf{\theta }}\left( z_{%
\mathbf{\theta }i}\right) \frac{\partial }{\partial \mathbf{\theta }}\hat{%
\varphi}_{\mathbf{\theta }}\left( z_{\mathbf{\theta }i}\right) -\tilde{%
\varphi}_{\mathbf{\theta }}\left( z_{\mathbf{\theta }i}\right) \frac{%
\partial }{\partial \mathbf{\theta }}\tilde{\varphi}_{\mathbf{\theta }%
}\left( z_{\mathbf{\theta }i}\right) \right\} \frac{I_{i}}{\pi _{i}} \\
&&-N^{-1}\sum_{i\in U_{N}}\tilde{\varphi}_{\mathbf{\theta }}\left( z_{%
\mathbf{\theta }i}\right) \frac{\partial }{\partial \mathbf{\theta }}\tilde{%
\varphi}_{\mathbf{\theta }}\left( z_{\mathbf{\theta }i}\right) \left( \frac{%
I_{i}}{\pi _{i}}-1\right).
\end{eqnarray*}
According to Lemma \ref{LEM:phihat-phitilde}, the first and the
third terms are of the order
$O_{p}\left(J_{N}^{1+k}/n_{N}^{1/2}\right)$. Using similar
arguments as in Lemma 2 of \citep{REF:BO}, it can be shown that $\frac{\partial }{\partial \mathbf{\theta }}\tilde{%
\varphi}_{\mathbf{\theta }}\left( z_{\mathbf{\theta }i}\right)$ is
bounded if Assumption (A8) holds. the second and the fourth terms
are of the order $O_{p}\left(n_{N}^{-1/2}\right)$. Thus the
desired result in (\ref{EQ:Rhat-Rtilde-div}) holds.
\end{proof}

\begin{lemma}
\label{LEM:PthetadotPpnorm}Under Assumptions (A2), (A4) and (A5),
there exists a constant $C>0$ such that with $\xi$-probability 1
\begin{equation}
\sup\limits_{1\leq p\leq d}\sup\limits_{\mathbf{\theta }\in
S_{c}^{d-1}}\left\| \frac{\partial }{\partial \theta
_{p}}\mathbf{P}_{\mathbf{\theta}}\right\| _{\infty }\leq CJ_{N}.
\label{EQ:PthetadotPpnorm}
\end{equation}
\end{lemma}
\begin{proof}
Note that for any $1 \leq p\leq d$
\begin{equation}
\dot{\mathbf{B}}_{p}\equiv\frac{\partial }{\partial \theta _{p}}\mathbf{B}_{\mathbf{\theta}}=\left[ \left\{ B_{j,3}\left( U_{\mathbf{\theta }%
,i}\right) -B_{j+1,3}\left( U_{\mathbf{\theta },i}\right) \right\} \dot{F}%
_{d}\left( \mathbf{X}_{\mathbf{\theta },i}\right)
h^{-1}X_{i,p}\right] _{i=1,j=-3}^{n,\text{ }N},  \label{EQ:dotBp}
\end{equation}
where $h$ is the length of the neighboring knots. For any vector $\mathbf{a}\in R^{n}$, with probability $1$%
\[
\left\| n^{-1}\mathbf{B}_{\mathbf{\theta }}^{T}\mathbf{a}\right\|
_{\infty }\leq \left\| \mathbf{a}\right\| _{\infty
}\max\limits_{-3\leq j\leq N}\left|
n^{-1}\sum_{i=1}^{n}B_{j,4}\left( U_{\mathbf{\theta },i}\right)
\right| \leq Ch\left\| \mathbf{a}\right\| _{\infty },\mbox{\ }
\left\| n^{-1}\mathbf{\dot{B}}_{p}^{T}\mathbf{a}\right\| _{\infty
}
\]
\[
\leq \left\| \mathbf{a}\right\| _{\infty }\max\limits_{-3\leq
j\leq N}\left|
\frac{1}{nh}\sum_{i=1}^{n}\left\{ \left( B_{j,3}-B_{j+1,3}\right) \left( U_{%
\mathbf{\theta },i}\right) \right\} \dot{F}_{d}\left( \mathbf{X}_{\mathbf{%
\theta },i}\right) X_{i,p}\right| \leq C\left\| \mathbf{a}\right\|
_{\infty }.
\]
Thus
\begin{equation}
\sup\limits_{\mathbf{\theta }\in S_{c}^{d-1}}\left\| n^{-1}\mathbf{B}_{%
\mathbf{\theta }}^{T}\right\| _{\infty }\leq Ch,\rm{ }
\sup\limits_{1\leq
p\leq d}\sup\limits_{\mathbf{\theta }\in S_{c}^{d-1}}\left\| n^{-1}\mathbf{%
\dot{B}}_{p}^{T}\right\| _{\infty }\leq C,a.s..
\label{EQ:BthetadotBpnorm}
\end{equation}
Observing that
\begin{equation}
\mathbf{\dot{P}}_{p}=\left( \mathbf{\mathbf{I}}-\mathbf{P}_{\mathbf{\theta }%
}\right) \mathbf{\dot{B}}_{p}\left( \mathbf{B}_{\mathbf{\theta }}^{T}\mathbf{%
B}_{\mathbf{\theta }}\right) ^{-1}\mathbf{B}_{\mathbf{\theta }}^{T}+\mathbf{B%
}_{\mathbf{\theta }}\left( \mathbf{B}_{\mathbf{\theta }}^{T}\mathbf{B}_{%
\mathbf{\theta }}\right) ^{-1}\mathbf{\dot{B}}_{p}^{T}\left( \mathbf{I}-%
\mathbf{P}_{\mathbf{\theta }}\right),  \label{EQ:dotPp}
\end{equation}
one only needs to combine (\ref{EQ:ninvBB}),
(\ref{EQ:BthetadotBpnorm}), and (\ref{EQ:dotPp}) to prove
(\ref{EQ:PthetadotPpnorm}).
\end{proof}

\section*{Appendix B. Proof of Theorems \ref{THM:theta consistent}-\ref{THM:normality}}
\renewcommand{\thetheorem}{{\sc B.\arabic{theorem}}} \renewcommand{%
\theproposition}{{\sc B.\arabic{proposition}}}
\renewcommand{\thelemma}{{\sc
B.\arabic{lemma}}} \renewcommand{\thecorollary}{{\sc
B.\arabic{corollary}}} \renewcommand{\theequation}{B.\arabic{equation}} %
\renewcommand{\thesubsection}{{B.\arabic{subsection}}} %
\setcounter{equation}{0} \setcounter{lemma}{0} \setcounter{proposition}{0} %
\setcounter{theorem}{0}
\setcounter{subsection}{0}\setcounter{corollary}{0}

\subsection{Proof of Theorem \ref{THM:theta consistent}}
Let $\left( \Omega ,\mathcal{A},\mathcal{P}\right) $ be the design
probability space with respect to the sampling design measure. By
Lemma \ref {LEM:Rtilde-Rhat}, for any $\delta >0$ and $\omega \in
\Omega $, there
exists an integer $N_{0}\left( \omega \right) $, such that when $%
N>N_{0}\left( \omega \right) $, $\hat{R}\left( \tilde{\mathbf{\theta }}%
,\omega \right) -\tilde{R}\left( \tilde{\mathbf{\theta }}\right)
<$ $\delta /2$. Note that
$\mathbf{\hat{\theta}}=\mathbf{\hat{\theta}}\left( \omega \right)
$ is the minimizer of $\hat{R}\left( \mathbf{\theta },\omega
\right) $, so $\hat{R}\left( \mathbf{\hat{\theta}}\left( \omega
\right) ,\omega \right) -\tilde{R}\left( \tilde{\mathbf{\theta
}}\right) <\delta /2$. Using
Lemma \ref{LEM:Rtilde-Rhat} again, there exists $N_{1}\left( \omega \right) $%
, such that when $N>N_{1}\left( \omega \right) $, $\tilde{R}\left( \mathbf{%
\hat{\theta}}\left( \omega \right) \right) -\hat{R}\left( \mathbf{\hat{\theta%
}}\left( \omega \right) ,\omega \right) <$ $\delta /2$. Thus, when
$N>\max \left( N_{0}\left( \omega \right) ,N_{1}\left( \omega
\right) \right) $,
\[
\tilde{R}\left( \mathbf{\hat{\theta}}\left( \omega \right) \right) -\tilde{R}%
\left( \tilde{\mathbf{\theta }}\right) <\delta /2+\hat{R}\left( \mathbf{\hat{%
\theta}}\left( \omega \right) ,\omega \right) -\tilde{R}\left( \tilde{%
\mathbf{\theta }}\right) <\delta /2+\delta /2=\delta .
\]
By Assumption (A7), for any $\varepsilon >0$, if $\tilde{R}\left( \mathbf{%
\hat{\theta}}\left( \omega \right) ,\omega \right) -\tilde{R}\left( \tilde{%
\mathbf{\theta }}\right) <\delta $, then one would have $\left\| \mathbf{%
\hat{\theta}}\left( \omega \right)-\tilde{\mathbf{\theta
}}\right\| _{2}<\varepsilon $ for $N$ large enough, which is true
for any $\omega $, and the strong consistency holds.

Next, note that
\[
\left. \frac{\partial \hat{R}\left( \mathbf{\theta }\right)
}{\partial
\mathbf{\theta }}\right| _{\mathbf{\theta =\hat{\theta}}}-\left. \frac{%
\partial \hat{R}\left( \mathbf{\theta }\right) }{\partial \mathbf{\theta }}%
\right| _{\mathbf{\theta =\tilde{\theta}}}=\left. \frac{\partial ^{2}\hat{R}%
\left( \mathbf{\theta }\right) }{\partial \mathbf{\theta }\partial \mathbf{%
\theta }^{T}}\right| _{\mathbf{\theta =\bar{\theta}}}\left( \mathbf{\hat{%
\theta}-\tilde{\theta}}\right) ,
\]
with
$\mathbf{\bar{\theta}}=\mathbf{t\hat{\theta}+(I-t)\tilde{\theta}}$.
So
\[
\mathbf{\hat{\theta}-\tilde{\theta}=-}\left( \left. \frac{\partial ^{2}\hat{R%
}\left( \mathbf{\theta }\right) }{\partial \mathbf{\theta }\partial \mathbf{%
\theta }^{T}}\right| _{\mathbf{\theta =\bar{\theta}}}\right)
^{-1}\left.
\frac{\partial \hat{R}\left( \mathbf{\theta }\right) }{\partial \mathbf{%
\theta }}\right| _{\mathbf{\theta =\tilde{\theta}}},
\]
where according to (\ref{EQ:Rhat-Rtilde-div}) and the above
consistency result of $\mathbf{\hat{\theta}}$, one has
\[
\lim_{N\rightarrow \infty }\left. \frac{\partial ^{2}\hat{R}\left( \mathbf{%
\theta }\right) }{\partial \mathbf{\theta }\partial \mathbf{\theta }^{T}}%
\right| _{\mathbf{\theta =\bar{\theta}}}\rightarrow \left.
\frac{\partial
^{2}\tilde{R}\left( \mathbf{\theta }\right) }{\partial \mathbf{\theta }%
\partial \mathbf{\theta }^{T}}\right| _{\mathbf{\theta =\tilde{\theta}}}
\]
in probability $p$, and by (\ref{EQ:Rhat-Rtilde-div}) in Lemma
\ref {LEM:Rtilde-Rhat}, one has
\[
\left| \left. \frac{\partial \hat{R}\left( \mathbf{\theta }\right) }{%
\partial \mathbf{\theta }}\right| _{\mathbf{\theta =\tilde{\theta}}}\right|
\leq \sup_{\mathbf{\theta }\in S_{c}^{d-1}}\left| \frac{\partial \hat{R}%
\left( \mathbf{\theta }\right) }{\partial \mathbf{\theta
}}-\frac{\partial \tilde{R}\left( \mathbf{\theta }\right)
}{\partial \mathbf{\theta }}\right| =O_{p}\left(
J_{N}/n_{N}^{1/2}\right) .
\]
Thus $\left\|
\mathbf{\hat{\theta}-\tilde{\theta}}\right\|_{\infty} =O_{p}\left(
J_{N}/n_{N}^{1/2}\right) $ by Assumption (A8).

\subsection{Proof of Theorem \ref{THM:ADU}}
\begin{lemma}
\label{LEM:Epmtilde-hat} Under Assumptions (A1)-(A5) and (A7) one
has
\[
\lim_{N\rightarrow \infty }\frac{1}{N}E_{p}\left[ \sum_{i\in
U_{N}}\left( \tilde{m}_{i}-\hat{m}_{i}\right) ^{2}\right] =0,
\]
where $\tilde{m}_{i}$ and $\hat{m}_{i}$ are defined in
(\ref{DEF:mitilde}) and (\ref{DEF:mihat}).
\end{lemma}
\begin{proof}
Let $\hat{\widetilde{m}}_{i}=e_{i}^{T}\mathbf{B}_{\mathbf{%
\hat{\theta}}}\left( \mathbf{B}_{\mathbf{\hat{\theta}}}^{T}\mathbf{B}_{%
\mathbf{\hat{\theta}}}\right) ^{-1}\mathbf{B}_{\mathbf{\hat{\theta}}}^{T}%
\mathbf{y}$, then one can write
\begin{eqnarray*}
&&\sum_{i\in U_{N}}\left( \tilde{m}_{i}-\hat{m}_{i}\right)
^{2}=\sum_{i\in U_{N}}\left(
\tilde{m}_{i}-\hat{\widetilde{m}}_{i}\right) ^{2}+\sum_{i\in
U_{N}}\left( \hat{\widetilde{m}}_{i}-\hat{m}_{i}\right) ^{2} \\
&&+2\sum_{i\in U_{N}}\left(
\tilde{m}_{i}-\hat{\widetilde{m}}_{i}\right) \left(
\hat{\widetilde{m}}_{i}-\hat{m}_{i}\right) .
\end{eqnarray*}
According to Lemma \ref{LEM:phihat-phitilde},
$\frac{1}{N}E_{p}\left[ \sum_{i\in U_{N}}\left(
\hat{\widetilde{m}}_{i}-\hat{m}_{i}\right) ^{2}\right] \rightarrow
0$, so it suffices to show
\begin{equation}
\frac{1}{N}E_{p}\left[ \sum_{i\in U_{N}}\left( \tilde{m}_{i}-\hat{\widetilde{%
m}}_{i}\right) ^{2}\right] \rightarrow 0.
\label{EQ:mtildehat-mhat}
\end{equation}
Let $f\left( t\right)
=e_{i}^{T}\mathbf{P}_{t\mathbf{\hat{\theta}+}\left( 1-t\right)
\mathbf{\tilde{\theta}}}\mathbf{y}$, then
\[
\frac{df\left( t\right) }{dt}=e_{i}^{T}\sum_{q=1}^{d}\frac{\partial }{%
\partial \theta _{q}}\mathbf{P}_{t\mathbf{\hat{\theta}+}\left( 1-t\right)
\mathbf{\tilde{\theta}}}\left(
\hat{\theta}_{q}-\tilde{\theta}_{q}\right) \mathbf{y}.
\]
Therefore,
\begin{eqnarray*}
\tilde{m}_{i}-\hat{\widetilde{m}}_{i} &=&e_{i}^{T}\mathbf{B}_{\mathbf{\hat{%
\theta}}}\left( \mathbf{B}_{\mathbf{\hat{\theta}}}^{T}\mathbf{B}_{\mathbf{%
\hat{\theta}}}\right) ^{-1}\mathbf{B}_{\mathbf{\hat{\theta}}}^{T}\mathbf{y-e}%
_{i}^{T}\mathbf{B}_{\tilde{\mathbf{\theta }}}\left( \mathbf{B}_{\tilde{%
\mathbf{\theta }}}^{T}\mathbf{B}_{\tilde{\mathbf{\theta }}}\right) ^{-1}%
\mathbf{B}_{\tilde{\mathbf{\theta }}}^{T}\mathbf{y} \\
&=&f\left( 1\right) -f\left( 0\right)
=e_{i}^{T}\sum_{q=1}^{d}\frac{\partial }{\partial \theta
_{q}}\mathbf{P}_{t^{*}\mathbf{\hat{\theta}+}\left(
1-t^{*}\right) \mathbf{\tilde{\theta}}}\left( \hat{\theta}_{q}-\tilde{\theta}%
_{q}\right) \mathbf{y},
\end{eqnarray*}
where $t^{*}\in (0,1)$. Thus
\[
\frac{1}{N}E_{p}\left[ \sum_{i\in U_{N}}\left( \tilde{m}_{i}-\hat{\widetilde{%
m}}_{i}\right) ^{2}\right] =\frac{1}{N}\sum_{i\in
U_{N}}E_{p}\left\{
e_{i}^{T}\sum_{q=1}^{d}\frac{\partial }{\partial \theta _{q}}\mathbf{P}%
_{t^{*}\mathbf{\hat{\theta}}+\left( 1-t^{*}\right) \mathbf{\tilde{\theta}}%
}\left( \hat{\theta}_{q}-\tilde{\theta}_{q}\right)
\mathbf{y}\right\} ^{2}.
\]
Note that according to Theorem \ref{THM:theta consistent}, with
$p$-probability 1,
\[
\frac{\partial }{\partial \theta _{q}}\mathbf{P}_{t^{*}\mathbf{\hat{\theta}+}%
\left( 1-t^{*}\right) \mathbf{\tilde{\theta}}}\rightarrow \frac{\partial }{%
\partial \theta _{q}}\mathbf{P}_{\mathbf{\tilde{\theta}}},
\]
and $\left\| \mathbf{\hat{\theta}-\tilde{\theta}}\right\|_{\infty}
=O_{p}\left( J_{N}/n_{N}^{1/2}\right)$. By Lemma
\ref{LEM:PthetadotPpnorm}, there exists a positive constant
$C_{0}$ such that $\sup_{1\leq k\leq d}\sup_{\mathbf{\theta }\in
S_{c}^{d-1}}\left\| \frac{\partial }{\partial \theta
_{k}}\mathbf{P}_{\mathbf{\theta}}\right\| _{\infty }\leq
C_{0}J_{N}$ with $\xi$-probability 1. Thus
(\ref{EQ:mtildehat-mhat}) follows directly from the above
arguments and Assumption (A4). Hence the result.
\end{proof}

Note that
\[
\hat{t}_{y,\text{diff}}-t_{y}=\sum_{i\in U_{N}}\left( y_{i}-\tilde{m}%
_{i}\right) \left( \frac{I_{i}}{\pi _{i}}-1\right) +\sum_{i\in
U_{N}}\left( \hat{m}_{i}-\tilde{m}_{i}\right) \left(
1-\frac{I_{i}}{\pi _{i}}\right).
\]
Then
\begin{eqnarray}
&&\left. E_{p}\left|
\frac{\hat{t}_{y,\text{diff}}-t_{y}}{N}\right| \leq
E_{p}\left| \sum_{i\in U_{N}}\left( y_{i}-\tilde{m}_{i}\right) \left( \frac{%
I_{i}}{\pi _{i}}-1\right) \right| \right.  \nonumber \\
&&+\left\{ E_{p}\left[ \sum_{i\in U_{N}}\frac{\left( \hat{m}_{i}-\tilde{m}%
_{i}\right) ^{2}}{N}\right] E_{p}\left[ \sum_{i\in
U_{N}}\frac{\left( 1-I_{i}/\pi _{i}\right) ^{2}}{N}\right]
\right\} ^{1/2} \label{EQ:ty_tilde-ty}
\end{eqnarray}
According to the definition of (\ref{DEF:Rtilde}), under
Assumptions (A1)-(A4), one has
\[
\limsup_{N\rightarrow \infty }\frac{1}{N}\sum_{i\in U_{N}}\left( y_{i}-%
\tilde{m}_{i}\right) ^{2}<\infty .
\]
Following the same argument of Theorem 1 in \citep{REF:BO}, the
first term on the right hand side of (\ref{EQ:ty_tilde-ty})
converges to zero as $N\rightarrow \infty $. For the second term,
Assumption (A5) implies that
\[
E_{p}\left[ \sum_{i\in U_{N}}\frac{\left( 1-I_{i}/\pi _{i}\right) ^{2}}{N}%
\right] =\sum_{i\in U_{N}}\frac{\pi _{i}\left( 1-\pi _{i}\right)
}{N\pi _{i}^{2}}\leq \frac{1}{\lambda }.
\]
According to Lemma \ref{LEM:Epmtilde-hat},
\[
\lim_{N\rightarrow \infty }\frac{1}{N}\sum_{i\in U_{N}}E_{p}\left[
\left( \hat{m}_{i}-\tilde{m}_{i}\right) ^{2}\right] \rightarrow
0\text{ with }\xi \text{-probability }1,
\]
and the result follows from the Markov's inequality.

\subsection{Proof of Theorem \ref{THM:normality}}

The next lemma is to derive the asymptotic mean squared error of
the proposed spline estimator in (\ref{DEF:tydiffhat}).
\begin{lemma}
\label{THM:variance}Under Assumptions (A1)-(A5) and (A7)
\begin{equation}
n_{N}E_{p}\left( \frac{\hat{t}_{y,\text{diff}}-t_{y}}{N}\right) ^{2}=%
\frac{n_{N}}{N^{2}}\sum_{i,j\in U_{N}}\left(
y_{i}-\tilde{m}_{i}\right)
\left( y_{j}-\tilde{m}_{i}\right) \left( \frac{\pi _{ij}-\pi _{i}\pi _{j}}{%
\pi _{i}\pi _{j}}\right) +o\left( 1\right) .  \label{DEF:AMSE}
\end{equation}
\end{lemma}
\begin{proof}
Note that
\[
\frac{\hat{t}_{y,\text{diff}}-t_{y}}{N}=\sum_{i\in U_{N}}\frac{y_{i}-%
\tilde{m}_{i}}{N}\left( \frac{I_{i}}{\pi _{i}}-1\right) +\sum_{i\in U_{N}}%
\frac{\hat{m}_{i}-\tilde{m}_{i}}{N}\left( \frac{I_{i}}{\pi
_{i}}-1\right).
\]
Let
\[
a_{N}=n_{N}^{1/2}\sum_{i\in U_{N}}\frac{y_{i}-\tilde{m}_{i}}{N}\left( \frac{%
I_{i}}{\pi _{i}}-1\right) ,b_{N}=n_{N}^{1/2}\sum_{i\in U_{N}}\frac{\hat{m}%
_{i}-\tilde{m}_{i}}{N}\left( 1-\frac{I_{i}}{\pi _{i}}\right) ,
\]
then
\begin{eqnarray*}
&& E_{p}\left[ a_{N}^{2}\right] =\frac{n_{N}}{N^{2}}\sum_{i,j\in
U_{N}}\left( y_{i}-\tilde{m}_{i}\right) \left(
y_{j}-\tilde{m}_{i}\right)
\left( \frac{\pi _{ij}-\pi _{i}\pi _{j}}{\pi _{i}\pi _{j}}\right) \\
&\leq &\left( \frac{1}{\lambda }+\frac{n_{N}\max_{i,j\in
U_{N},i\neq
j}\left| \pi _{ij}-\pi _{i}\pi _{j}\right| }{\lambda ^{2}}\right) \frac{1}{N}%
\sum_{i\in U_{N}}\left( y_{i}-\tilde{m}_{i}\right) ^{2} < \infty ,
\end{eqnarray*}
\begin{eqnarray*}
E_{p}\left[ b_{N}^{2}\right]
&=&\frac{n_{N}}{N^{2}}E_{p}\sum_{i,j\in
U_{N}}\left( \hat{m}_{i}-\tilde{m}_{i}\right) \left( \hat{m}_{j}-\tilde{m}%
_{j}\right) \left( 1-\frac{I_{i}}{\pi _{i}}\right) \left(
1-\frac{I_{j}}{\pi
_{j}}\right) \\
&\leq &\left( \frac{1}{\lambda }+\frac{n_{N}\max_{i,j\in
U_{N},i\neq
j}\left| \pi _{ij}-\pi _{i}\pi _{j}\right| }{\lambda ^{2}}\right) \frac{1}{N}%
E_{p}\left[ \left\{ \hat{m}_{i}-\tilde{m}_{i}\right\} ^{2}\right]
.
\end{eqnarray*}
By Lemma \ref{LEM:Epmtilde-hat}, one has $E_{p}\left[
b_{N}^{2}\right] =o\left( 1\right) $ and Cauchy-Schwartz
inequality implies $E_{p}\left[ a_{n}b_{n}\right] =0$. Therefore
\[
n_{N}E_{p}\left( \frac{\hat{t}_{y,\text{diff}}-t_{y}}{N}\right)
^{2}=E_{p}\left[ a_{n}^{2}\right] +2E_{p}\left[ a_{n}b_{n}\right]
+E_{p}\left[ b_{n}^{2}\right] =E_{p}\left[ a_{n}^{2}\right]
+o\left( 1\right) .
\]
Thus the desired result holds.
\end{proof}

Denote by
\[
\text{AMSE}\left( N^{-1}\hat{t}_{y,\text{diff}}\right) =\frac{1}{N^{2}}%
\sum_{i,j\in U_{N}}\left( y_{i}-m_{i}\right) \left(
y_{j}-m_{j}\right) \left( \frac{\pi _{ij}-\pi _{i}\pi _{j}}{\pi
_{i}\pi _{j}}\right)
\]
as the asymptotic mean squared error in (\ref{DEF:AMSE}). The next
result
shows that it can be estimated consistently by $\hat{V}\left( N^{-1}\hat{t}%
_{y,\text{diff}}\right)$ in (\ref{DEF:Vhat}).

\begin{lemma}
\label{THM:Vhat-AMSE} Under (A1)-(A7), one has
\[
\lim_{N\rightarrow \infty }n_{N}E_{p}\left| \hat{V}\left( N^{-1}\hat{t}_{y,%
\text{diff}}\right) -\text{AMSE}\left(
N^{-1}\hat{t}_{y,\text{diff}}\right) \right| =0.
\]
\end{lemma}
\begin{proof}
Denote
\begin{eqnarray*}
S_{1} &=&\frac{1}{N^{2}}\sum_{i,j\in U_{N}}\left(
y_{i}-\tilde{m}_{i}\right)
\left( y_{j}-\tilde{m}_{j}\right) \left( \frac{\pi _{ij}-\pi _{i}\pi _{j}}{%
\pi _{i}\pi _{j}}\right) \frac{I_{i}I_{j}}{\pi _{ij}}, \\
S_{2} &=&\frac{2}{N^{2}}\sum_{i,j\in U_{N}}\left(
y_{i}-\tilde{m}_{i}\right) \left( \tilde{m}_{j}-\hat{m}_{j}\right)
\left( \frac{\pi _{ij}-\pi _{i}\pi
_{j}}{\pi _{i}\pi _{j}}\right) \frac{I_{i}I_{j}}{\pi _{ij}}, \\
S_{3} &=&\frac{1}{N^{2}}\sum_{i,j\in U_{N}}\left( \tilde{m}_{i}-\hat{m}%
_{i}\right) \left( \tilde{m}_{j}-\hat{m}_{j}\right) \left(
\frac{\pi
_{ij}-\pi _{i}\pi _{j}}{\pi _{i}\pi _{j}}\right) \frac{I_{i}I_{j}}{\pi _{ij}}%
,
\end{eqnarray*}
then
\begin{equation}
\hat{V}\left( N^{-1}\hat{t}_{y,\text{diff}}\right) -\text{AMSE}\left( N^{-1}%
\hat{t}_{y,\text{diff}}\right) =S_{1}-\text{AMSE}\left( N^{-1}\hat{t}_{y,%
\text{diff}}\right) +S_{2}+S_{3}.  \label{EQ:Vhat-AMSE}
\end{equation}
For the first term $S_{1}$, one has
\begin{eqnarray*}
&&n_{N}E_{p}\left| S_{1}-\text{AMSE}\left( N^{-1}\hat{t}_{y,\text{diff}%
}\right) \right| \\
&=&\frac{n_{N}}{N^{2}}E_{p}\left| \sum_{i,j\in U_{N}}\left( y_{i}-\tilde{m}%
_{i}\right) \left( y_{j}-\tilde{m}_{j}\right) \left( \frac{\pi
_{ij}-\pi
_{i}\pi _{j}}{\pi _{i}\pi _{j}}\right) \frac{I_{i}I_{j}-\pi _{ij}}{\pi _{ij}}%
\right| \\
&\leq& \frac{n_{N}}{N^{2}}\left\{ E_{p}\left[ \sum_{i,j\in
U_{N}}\left(
y_{i}-\tilde{m}_{i}\right) \left( y_{j}-\tilde{m}_{j}\right) \left( \frac{%
\pi _{ij}-\pi _{i}\pi _{j}}{\pi _{i}\pi _{j}}\right)
\frac{I_{i}I_{j}-\pi _{ij}}{\pi _{ij}}\right] ^{2}\right\} ^{1/2},
\end{eqnarray*}
and
\begin{eqnarray*}
&&\frac{n_{N}^{2}}{N^{4}}E_{p}\left[ \sum_{i,j\in U_{N}}\left( y_{i}-\tilde{m%
}_{i}\right) \left( y_{j}-\tilde{m}_{j}\right) \left( \frac{\pi
_{ij}-\pi
_{i}\pi _{j}}{\pi _{i}\pi _{j}}\right) \frac{I_{i}I_{j}-\pi _{ij}}{\pi _{ij}}%
\right] ^{2} \\
&=&\frac{n_{N}^{2}}{N^{4}}\sum_{i,j,k,l\in U_{N}}\left( y_{i}-\tilde{m}%
_{i}\right) \left( y_{j}-\tilde{m}_{j}\right) \left( y_{k}-\tilde{m}%
_{k}\right) \left( y_{l}-\tilde{m}_{l}\right) \\
&&\times \left( \frac{\pi _{ij}-\pi _{i}\pi _{j}}{\pi _{i}\pi
_{j}}\right) \left( \frac{\pi _{kl}-\pi _{k}\pi _{l}}{\pi _{k}\pi
_{l}}\right)
E_{p}\left( \frac{I_{i}I_{j}-\pi _{ij}}{\pi _{ij}}\frac{I_{k}I_{l}-\pi _{kl}%
}{\pi _{kl}}\right) ,
\end{eqnarray*}
which can be represented as
\begin{eqnarray*}
&&\frac{n_{N}^{2}}{N^{4}}\sum_{i,k\in U_{N}}\left( y_{i}-\tilde{m}%
_{i}\right) ^{2}\left( y_{k}-\tilde{m}_{k}\right) ^{2}\left( \frac{1-\pi _{i}%
}{\pi _{i}}\right) \left( \frac{1-\pi _{k}}{\pi _{k}}\right) \left( \frac{%
\pi _{ik}-\pi _{i}\pi _{k}}{\pi _{ik}}\right) \\
&&+\frac{2n_{N}^{2}}{N^{4}}\sum_{i\in U_{N}}\sum_{k\neq l}\left( y_{i}-%
\tilde{m}_{i}\right) ^{2}\left( y_{k}-\tilde{m}_{k}\right) \left( y_{l}-%
\tilde{m}_{l}\right) \left( \frac{1-\pi _{i}}{\pi _{i}}\right) \left( \frac{%
\pi _{kl}-\pi _{k}\pi _{l}}{\pi _{kl}}\right) \\
&&\times E_{p}\left( \frac{I_{i}-\pi _{i}}{\pi
_{i}}\frac{I_{k}I_{l}-\pi
_{kl}}{\pi _{kl}}\right) \\
&&+\frac{n_{N}^{2}}{N^{4}}\sum_{i\neq j,k\neq l}\left( y_{i}-\tilde{m}%
_{i}\right) \left( y_{j}-\tilde{m}_{j}\right) \left( y_{k}-\tilde{m}%
_{k}\right) \left( y_{l}-\tilde{m}_{l}\right) \\
&&\times \left( \frac{\pi _{ij}-\pi _{i}\pi _{j}}{\pi _{i}\pi
_{j}}\right) \left( \frac{\pi _{kl}-\pi _{k}\pi _{l}}{\pi _{k}\pi
_{l}}\right)
E_{p}\left( \frac{I_{i}I_{j}-\pi _{ij}}{\pi _{ij}}\frac{I_{k}I_{l}-\pi _{kl}%
}{\pi _{kl}}\right) \\
&\equiv &s_{1N}+s_{2N}+s_{3N}.
\end{eqnarray*}
Now
\begin{eqnarray*}
s_{1N} &\leq &n_{N}^{2}\sum_{i\in U_{N}}\frac{\left( y_{i}-\tilde{m}%
_{i}\right) ^{4}}{\lambda ^{3}N^{4}}+n_{N}^{2}\sum_{i,k\in U_{N}}\frac{%
\left( y_{i}-\tilde{m}_{i}\right) ^{2}\left(
y_{k}-\tilde{m}_{k}\right)
^{2}\left| \pi _{ik}-\pi _{i}\pi _{k}\right| }{\lambda ^{4}N^{4}} \\
&\leq &\left( \frac{n_{N}^{2}}{\lambda ^{3}N^{3}}+n_{N}^{2}\frac{%
\max_{i,k\in U_{N},i\neq k}\left| \pi _{ik}-\pi _{i}\pi _{k}\right| }{%
\lambda ^{4}N^{2}}\right) \sum_{i\in U_{N}}\frac{\left( y_{i}-\tilde{m}%
_{i}\right) ^{4}}{N},
\end{eqnarray*}
and $\lim \sup_{N\rightarrow \infty }%
\frac{1}{N}\sum_{i\in U_{N}}\left( y_{i}-\tilde{m}_{i}\right)
^{4}<\infty $. Thus $s_{1N}$ goes to zero as $N\rightarrow \infty
$. Next
\[
s_{3N}\leq \frac{\left( n_{N}\max_{i,k\in U_{N},i\neq k}\left| \pi
_{ik}-\pi _{i}\pi _{k}\right| \right) ^{2}}{\lambda ^{4}\lambda
^{*2}}
\]
\[
\times \frac{1}{N^{4}}\sum_{i\neq j,k\neq l}\left| \left( y_{i}-\tilde{m}%
_{i}\right) \left( y_{j}-\tilde{m}_{j}\right) \left( y_{k}-\tilde{m}%
_{k}\right) \left( y_{l}-\tilde{m}_{l}\right) \right| E_{p}\left| \frac{%
I_{i}I_{j}-\pi _{ij}}{\pi _{ij}}\frac{I_{k}I_{l}-\pi _{kl}}{\pi _{kl}}%
\right|
\]
\begin{eqnarray*}
&\leq &O\left( N^{-1}\right) +\frac{\left( n_{N}\max_{i,k\in
U_{N},i\neq k}\left| \pi _{ik}-\pi _{i}\pi _{k}\right| \right)
^{2}}{\lambda ^{4}\lambda
^{*2}} \\
&&\times \max_{\left( i,j,k,l\right) \in D_{4,N}}\left|
E_{p}\left[ \left( \frac{I_{i}I_{j}-\pi _{ij}}{\pi _{ij}}\right)
\left( \frac{I_{k}I_{l}-\pi
_{kl}}{\pi _{kl}}\right) \right] \right| \sum_{i\in U_{N}}\frac{\left( y_{i}-%
\tilde{m}_{i}\right) ^{4}}{N},
\end{eqnarray*}
which converges to zero as $N\rightarrow \infty $ by Assumption
(A6). As a
result of Cauchy-Schwartz inequality, one can show $s_{2N}$ goes to zero as $%
N\rightarrow \infty $. Therefore
\begin{equation}
n_{N}E_{p}\left| S_{1}-\text{AMSE}\left( N^{-1}\hat{t}_{y,\text{diff}%
}\right) \right| \rightarrow 0,\text{as }N\rightarrow \infty .
\label{EQ:S1}
\end{equation}
Next for $S_{2}$, by Lemma \ref{LEM:Epmtilde-hat}
\begin{eqnarray}
&&n_{N}E_{p}\left| \frac{2}{N^{2}}\sum_{i,j\in U_{N}}\left( y_{i}-\tilde{m}%
_{i}\right) \left( \tilde{m}_{j}-\hat{m}_{j}\right) \left(
\frac{\pi
_{ij}-\pi _{i}\pi _{j}}{\pi _{i}\pi _{j}}\right) \frac{I_{i}I_{j}}{\pi _{ij}}%
\right|  \nonumber \\
&\leq &\left( \frac{2n_{N}\max_{i,k\in U_{N},i\neq k}\left| \pi
_{ik}-\pi _{i}\pi _{k}\right| ^{2}}{\lambda ^{4}\lambda
^{*2}}+\frac{2n_{N}}{\lambda
^{2}N}\right) \sum_{i\in U_{N}}\frac{\left( y_{i}-\tilde{m}_{i}\right) ^{2}}{%
N}  \nonumber \\
&\times &\sum_{i\in U_{N}}\frac{E_{p}\left(
\tilde{m}_{i}-\hat{m}_{i}\right) ^{2}}{N}\rightarrow 0.
\label{EQ:S2}
\end{eqnarray}
For $S_{3}$, applying Lemma \ref{LEM:Epmtilde-hat} again, one has
\begin{eqnarray}
&& E_{p}\left| n_{N}S_{3}\right|=\frac{n_{N}}{N^{2}}E_{p}\left|
\sum_{i,j\in
U_{N}}\left( \tilde{m}_{i}-\hat{m}_{i}\right) \left( \tilde{m}_{j}-\hat{m}%
_{j}\right) \left( \frac{\pi _{ij}-\pi _{i}\pi _{j}}{\pi _{i}\pi _{j}}%
\right) \frac{I_{i}I_{j}}{\pi _{ij}}\right|  \nonumber \\
&\leq &\left( \frac{n_{N}\max_{i,j\in U_{N},i\neq j}\left| \pi
_{ij}-\pi
_{i}\pi _{j}\right| }{\lambda ^{2}\lambda ^{*}}+\frac{n_{N}}{\lambda ^{2}N}%
\right) \frac{1}{N}E_{p}\left[ \sum_{i\in U_{N}}\left( \tilde{m}_{i}-\hat{m}%
_{i}\right) ^{2}\right] \rightarrow 0  \label{EQ:S3}
\end{eqnarray}
The desired result follows from
(\ref{EQ:Vhat-AMSE})-(\ref{EQ:S3}).
\end{proof}

\textbf{Proof of Theorem \ref{THM:normality}:} By the proof of
Lemma \ref {THM:variance},
\begin{eqnarray*}
N^{-1}\left( \hat{t}_{y,\text{diff}}-t_{y}\right) &=&N^{-1}\left(
\sum_{i\in s}\frac{y_{i}-\hat{m}_{i}}{\pi _{i}}+\sum_{i\in
U_{N}}\hat{m}_{i}-\sum_{i\in
U_{N}}\frac{y_{i}I_{i}}{\pi _{i}}\right) \\
&=&N^{-1}\left( \tilde{t}_{y,\text{diff}}-t_{y}\right)
+o_{p}\left( n_{N}^{-1/2}\right) ,
\end{eqnarray*}
so the desired result follows from Lemma \ref{THM:Vhat-AMSE}.
\hfill $\square$

\end{document}